%% file: ms.tex
\documentclass[sigconf]{acmart}
\setcopyright{none}

\usepackage{booktabs} % For formal tables
\usepackage{amsmath}
\usepackage{array}
\usepackage{longtable}
\usepackage{makecell}
\usepackage{epstopdf}
\usepackage{float}
\usepackage{multirow}
\usepackage{balance}
\usepackage{bm}
\usepackage{algorithm}
\usepackage{algorithmic}
\usepackage{makecell}
\usepackage{xfrac}
\usepackage{mathtools}
\usepackage{pdfpages}
\usepackage{url}
\usepackage{graphicx}
\usepackage{caption}
\usepackage{subcaption}
\usepackage[flushleft]{threeparttable}

\DeclarePairedDelimiterX{\infdivx}[2]{(}{)}{%
  #1\;\delimsize\|\;#2%
}

\DeclareMathOperator*{\argmax}{arg\,max}
\newlength\myindent
\newcommand\bindent{%
  \begingroup
  \setlength{\itemindent}{\myindent}
  \addtolength{\algorithmicindent}{\myindent}
}
\newcommand\eindent{\endgroup}

\newcommand*\rotv{\rotatebox{90}}

\newcolumntype{L}{>{\centering\arraybackslash}m{2.2cm}}
\newcommand*\rfrac[2]{{}^{#1}\!/_{#2}}

\begin{document}

\title{HoloScope: Topology-and-Spike Aware Fraud Detection}
\author{Shenghua Liu,$^{1,2,*}$ Bryan Hooi,$^{2}$ Christos Faloutsos$^{2}$}
\authornote{The work was done when Shenghua Liu was a visiting researcher at
Carnegie Mellon University}
\affiliation{{{$^{1}$}CAS Key Laboratory of Network Data Science \& Technology,\\
Institute of Computing Technology, Chinese Academy of Sciences}\\
{{$^{2}$}Carnegie Mellon University}}
\email{liu.shengh@gmail.com, {bhooi, christos}@cs.cmu.edu}

\begin{abstract}
    \input{0-abs}

\end{abstract}

%Contrast Suspiciousness,

\maketitle

\input{1-intro}

\input{2-relwork}

\newtheorem{complexity}{complexity}
\newtheorem{axiom}{Axiom}
\newtheorem{trait}{Trait}
\newtheorem{informal problem}{Informal Problem}
\newtheorem{problem}{Problem}
\input{3-alg}
\input{4-exp}

\input{5-con}

\bibliographystyle{ACM-Reference-Format}
\balance
\bibliography{frauddetect}

\end{document}

%% file: 0-abs.tex
As online fraudsters invest more resources, including purchasing large pools of fake user accounts and dedicated IPs, fraudulent attacks become less obvious and their detection becomes increasingly challenging. Existing approaches such as average degree maximization suffer from the bias of including more nodes than necessary, resulting in lower accuracy and increased need for manual verification. Hence, we propose HoloScope, which uses information from graph topology and temporal spikes to more accurately detect groups of fraudulent users. In terms of graph topology, we introduce ``contrast suspiciousness,'' a dynamic weighting approach, which allows us to more accurately detect fraudulent blocks, particularly low-density blocks. In terms of temporal spikes, HoloScope takes into account the sudden bursts and drops of fraudsters' attacking patterns. In addition, we provide theoretical bounds for how much this increases the time cost needed for fraudsters to conduct adversarial attacks. Additionally, from the perspective of ratings, HoloScope incorporates the deviation of rating scores in order to catch fraudsters more accurately. Moreover, HoloScope has a concise framework and sub-quadratic time complexity, making the algorithm reproducible and scalable. Extensive experiments showed that HoloScope achieved significant accuracy improvements on synthetic and real data, compared with state-of-the-art fraud detection methods. 

% Fraud detection is facing more challenges when 
% the fraudsters invest more resources, like user
% accounts and dedicated IPs, to make their
% connections less dense than ever before.
% However, existing work that aimed at finding the connections of 
% maximum average degree density, 
% suffered a bias of
% including more nodes than the necessary,
% resulting in lower accuracy and intensive manual verification.
% Therefore, we propose HoloScope method to detect
% topology and spike suspiciousness simultaneously.
% First on the topology perspective, we introduce a ``contrast suspiciousness''
% to dynamically adjust the weights of nodes, which
% can achieve a better accuracy especially to detect a low-density
% fraudulent connections.
% Second on the timing, our HoloScope can honor the 
% sudden burst and drop of fraudsters' attacking time
% on the targets, which in fact increases the time cost of 
% the fraudsters' adversarial attacks by a theoretical
% bound. 
% On the rating perspective, the HoloScope
% considers the deviation of rating scores into 
% our dynamic contrast suspicious.
% Moreover, the HoloScope has a concise framework and sub-quadratic 
% time complexity, which make our algorithm reproducible and scalable.
% Overall, the HoloScope achieved significant accuracy 
% improvements on both synthetic and real data,
% compared with the state-of-the-art fraud detection
% methods.

%% file: 1-intro.tex
\section{Introduction}
Online fraud has become an increasingly serious problem due to the high profit it offers to fraudsters, which can be as much as \$5 million
from 300 million fake ``views'' per day,
according to a report of Methbot~\cite{Methbot} on
Dec 2016.
Meanwhile, to avoid detection, fraudsters can manipulate their geolocation, 
internet providers, and IP address, via large IP pools (852,992 dedicated IPs). 
Suppose a fraudster has $a$ accounts or IPs, and wants to rate or click $b$
objects (e.g., products) at least 200 times, as required by their customers.
The edge density of the fraudulent block is then 
% $\sfrac{(b \times 200)}{(a \times b)}$ $=$
% $\sfrac{200}{a}$.
$200\cdot b / (a \cdot b) = 200 / a$.
We thus see that with enough user accounts or IPs, 
the fraudsters can serve as many customers as they need while keeping density low. 
This presents a difficult challenge for most existing fraud detection methods. 
%\begin{equation}
%    \nonumber
%    density = \frac{b \times 200}{a \times b} = \frac{200}{a}.
%\end{equation}
% We can see that the volume density of fraudulent connections 
% is only related to the number of fraudsters,
% and they can serve as many customers as possible while keeping
% the same density.
%does not related to the number of objects, that is to say, 
%they can theoretically 
%fraud infinite number of objects, keeping a density in adjacency
%matrix.
% We call the density of fraudulent connections as
% ``fraudulent density''.
% Due to the high profit, 
% the fraudsters can invest more accounts and IPs
% to decrease the fraudulent density.
%investment return, which decreases the density.
%the controlled accounts or IPs have the property of
%an asset, which are considered as an investment in  
%finance. Fraudsters can enlarge the scale of controlled
%accounts and IPs gradually according to their revenue and 
%investment return, which decreases the density.
%For example, if having 2000 accounts or IPs, the density can keep 
%as low as 0.1; and moreover, if having more than 10,000, 
%the density will be 0.02.
% Therefore,
% plus the camouflage, such a low fraudulent density 
% challenges most of the existing fraud 
% detection methods.
% Directly detecting fraud by volume density makes the
% algorithms easily trapped into a trivial solution, namely,
% a $1\times 1$ block. 

Current dense block 
detection methods~\cite{charikar2000greedy,shin2016mzoom,shin2017dcube}  
maximize the arithmetic or geometric average degree.
We use ``fraudulent density'' to indicate the edge density that 
fraudsters create for the target objects.
However, those methods have 
a bias of including more nodes than necessary,
especially as the fraudulent density decreases, which we verified empirically.
This bias results in low precision,
which then requires intensive manual work to verify each user.
Fraudar~\cite{hooi2016fraudar} proposed an edge weighting scheme based on inverse logarithm 
of objects' degrees to reduce this bias, which was inspired by 
IDF~\cite{sparck1972statistical,robertson2004understanding}. 
However, their weighting scheme is fixed globally and affects
both suspicious and normal edges, lowering the precision of Fraudar,
which can be seen from results 
on semi-real (with injected labels) and real
data (see Fig.~\ref{figcmp}).
%In (a), the upper left rectangular is 
%the most suspicious and dense block. 
%Although comparing with the ``average degree'' method (red box),
%Fraudar (green box) improves the accuracy, but the accuracy 
%is still not significant.

%\subsection{Motivation}
\begin{figure*}
\centering
    \begin{subfigure}[b]{0.33\textwidth}
	\centering
	{\includegraphics[height=1.8in]{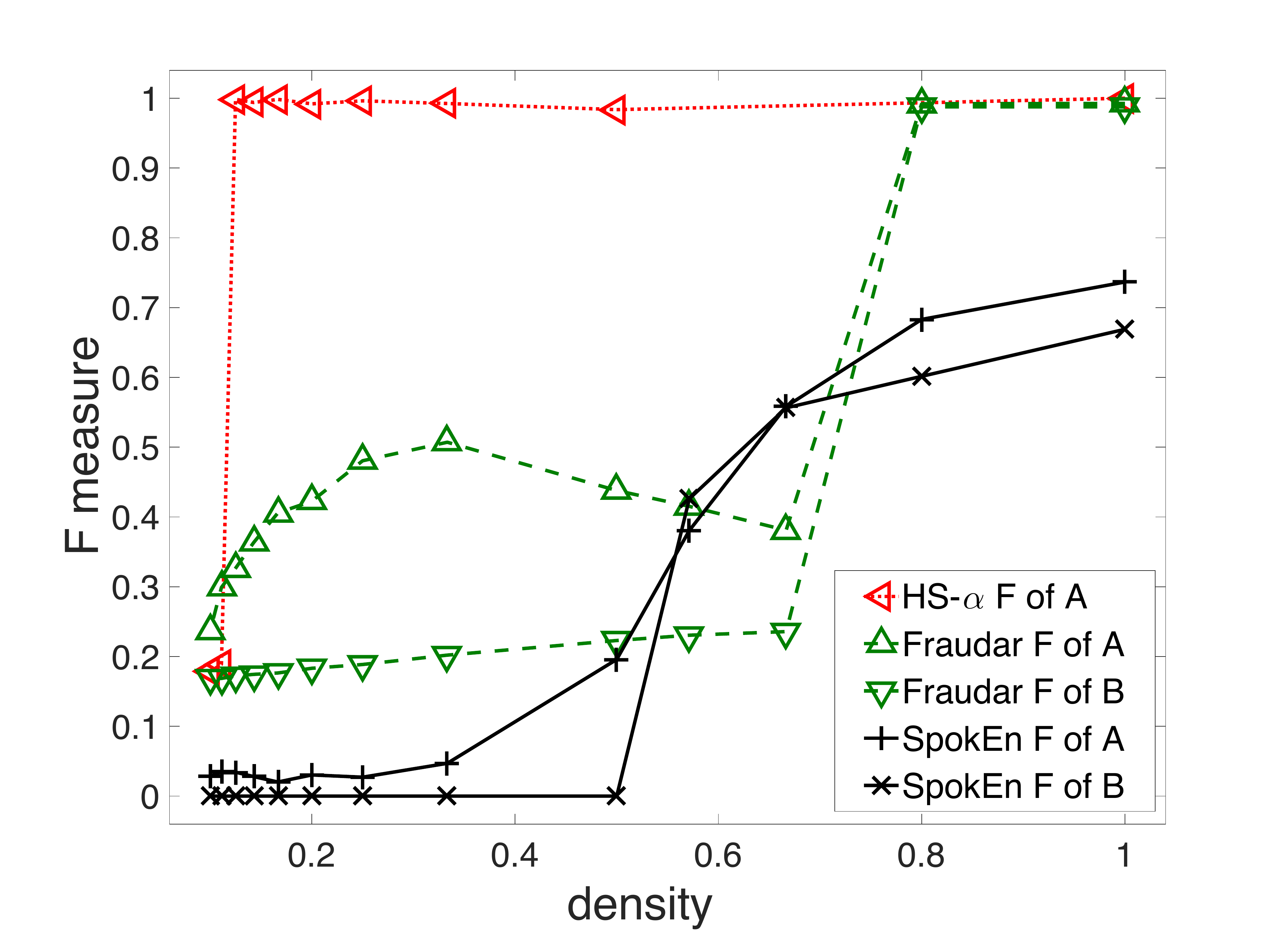}}
	\caption{HS-$\alpha$ using topology information outperforms
	competitors}
	\label{sfigconncmp}
    \end{subfigure}
    \begin{subfigure}[b]{0.33\textwidth}
	\centering
	{\includegraphics[height=1.8in]{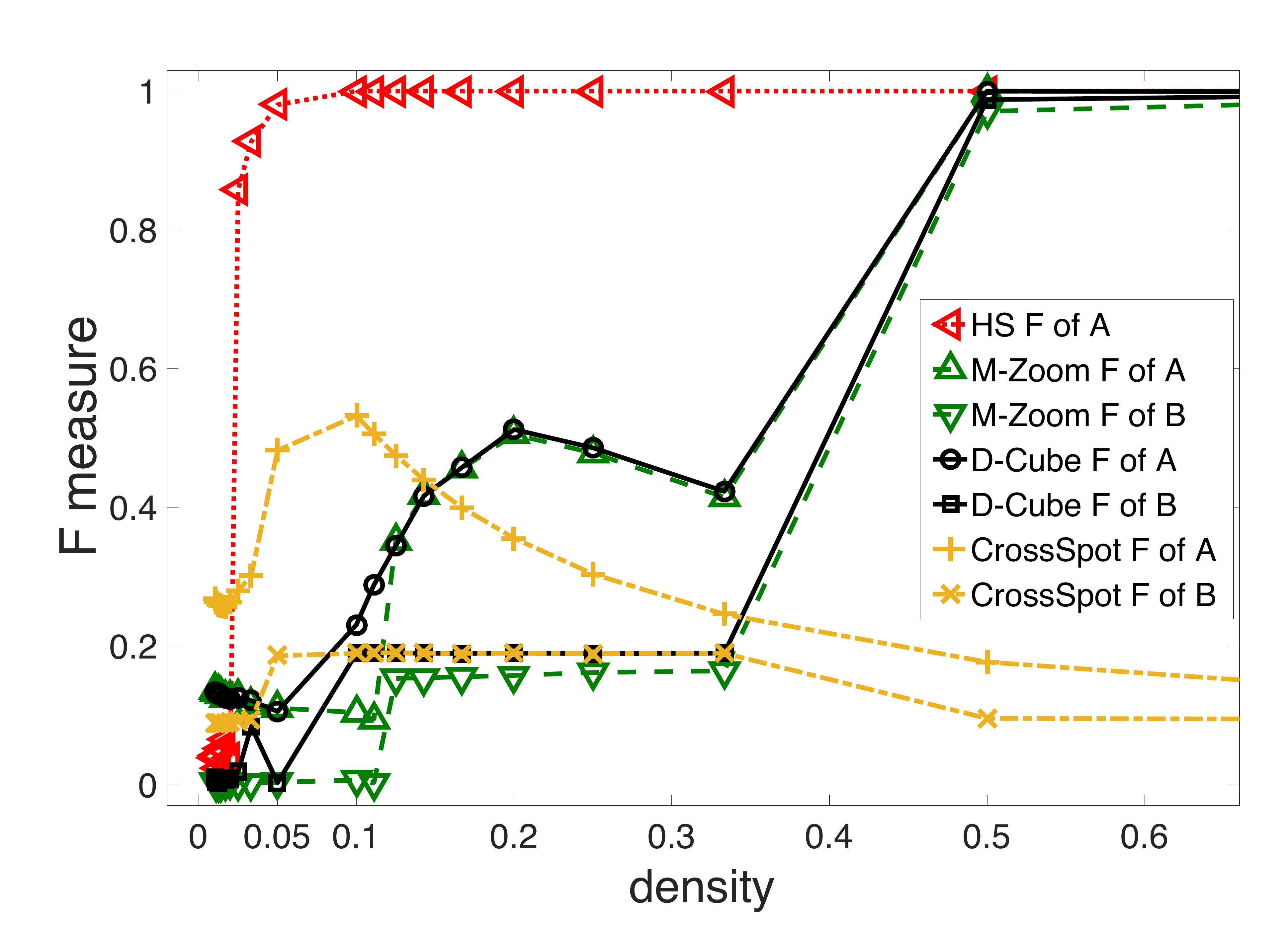}}
	\caption{HS using holistic attributes provides clear improvement, and
	performs the best.}
	\label{sfigallcmp}
    \end{subfigure}
    \begin{subfigure}[b]{0.3\textwidth}
	\centering
	\includegraphics[height=1.8in]{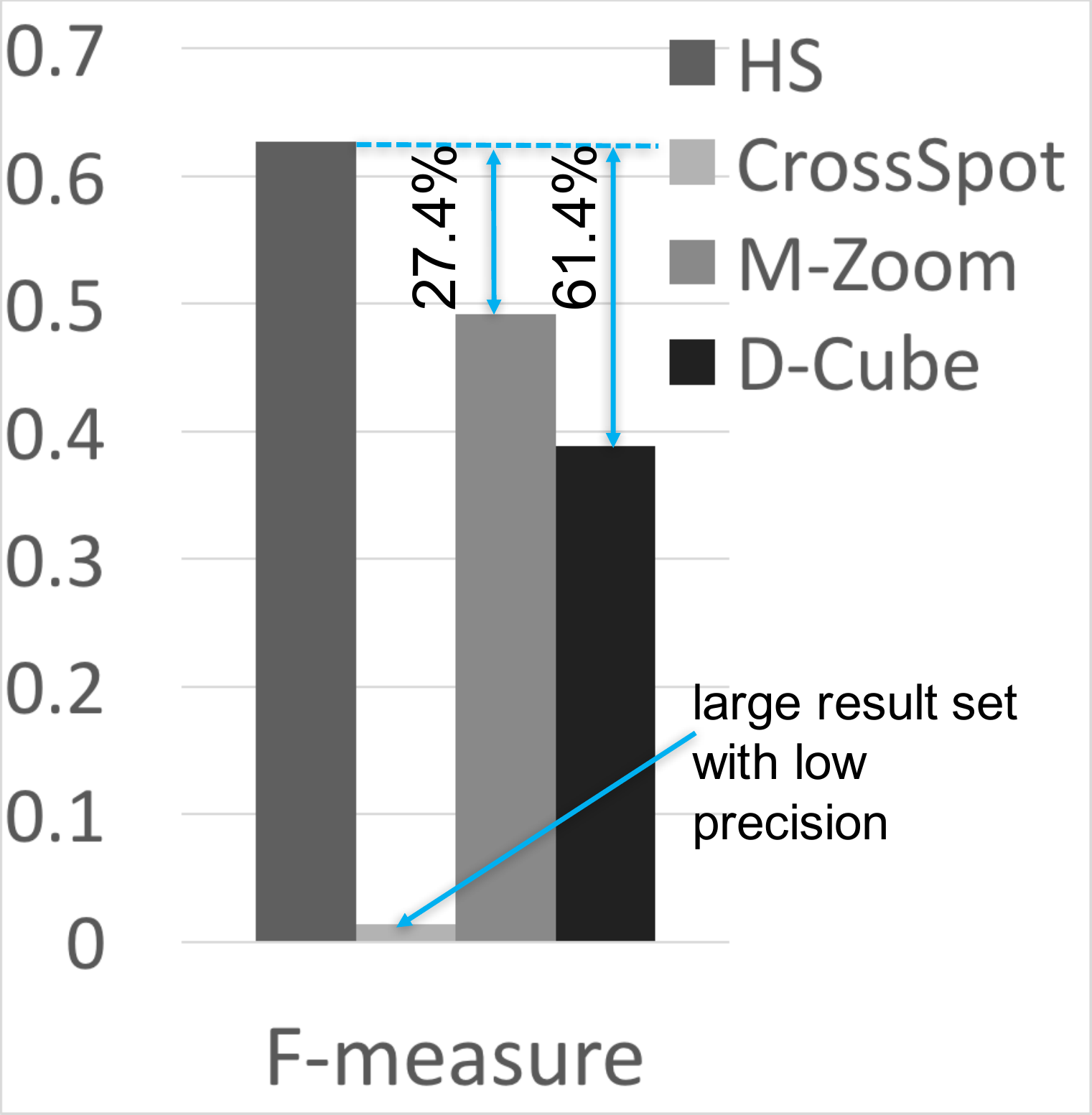}
	\caption{HS achieves the best F-measure, on real data from Sina Weibo}
	\label{sfigwbexp}
    \end{subfigure}
    \caption{(a) and (b) show experimental results on BeerAdvocate dataset. 
    HS-$\alpha$ and HS are both our methods, where the
    former only uses topology information. 
    We increase the \# of injected fraudsters from 200 to 2000 for HS-$\alpha$, and
    to 20000 for HS. The corresponding density is shown on the horizontal axis. 
    %Since HS-$\alpha$ and HS return suspiciousness scores, 
    %we measure accuracy using AUC (area under the ROC curve).
    (c) shows accuracy results on Sina Weibo, with ground truth labels.}
\label{figcmp}
\end{figure*}

\begin{figure}
    \centering
    \begin{subfigure}[b]{0.23\textwidth}
	\centering
	{\includegraphics[height=1.2in]{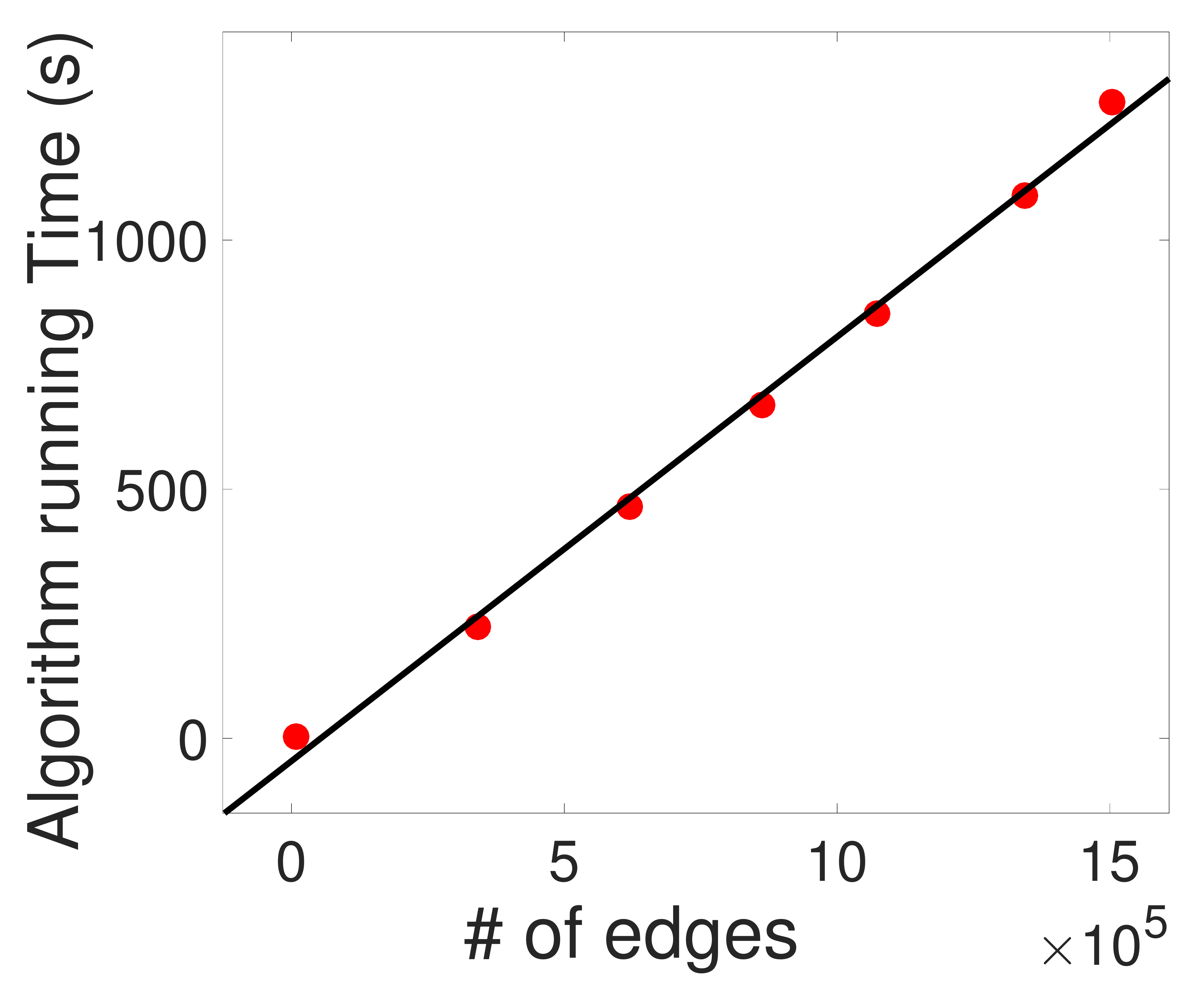}}
	\caption{BeerAdvocate dataset}
	%\label{sfigconncmp}
    \end{subfigure}
    \begin{subfigure}[b]{0.23\textwidth}
	\centering
	{\includegraphics[height=1.3in]{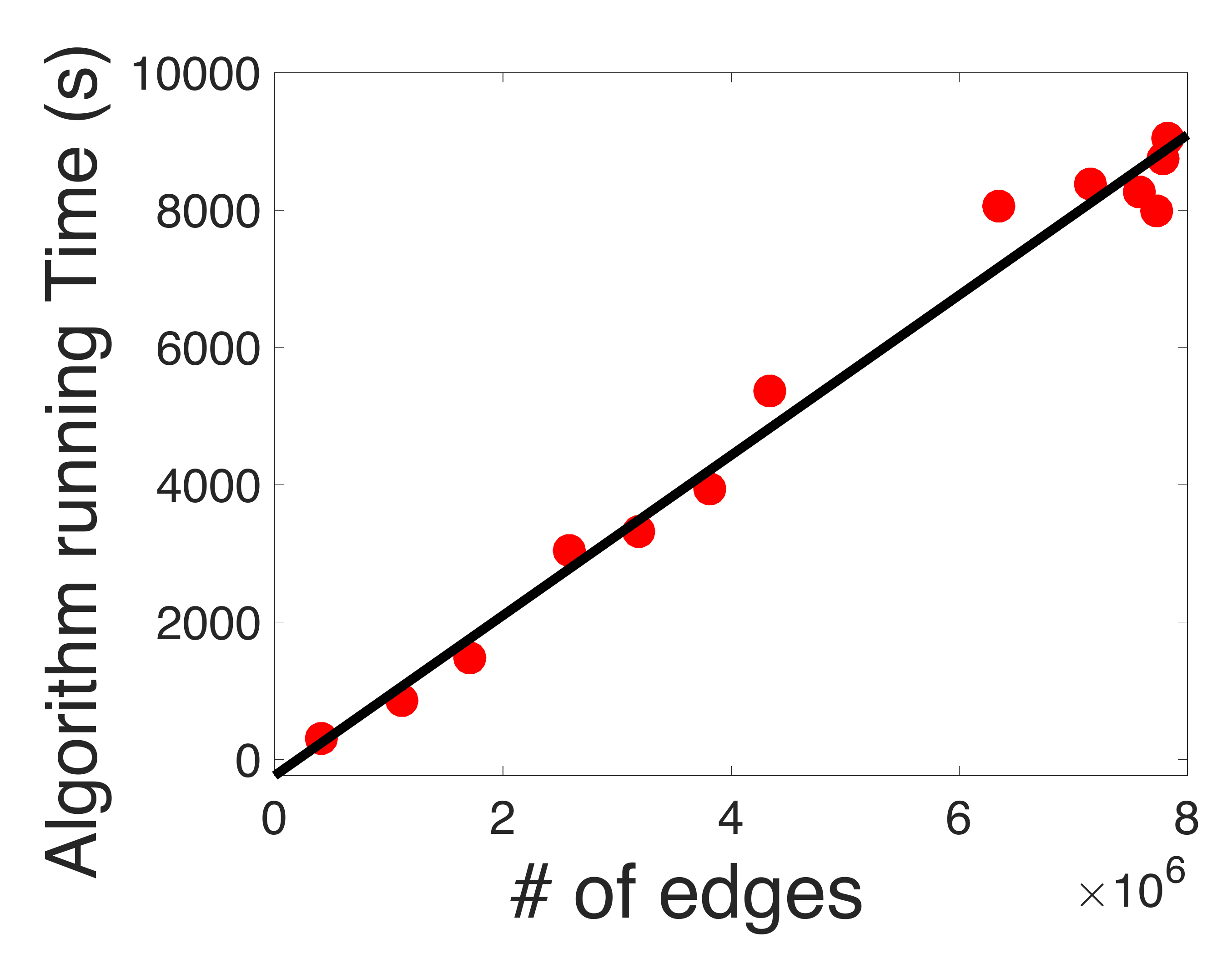}}
	\caption{Amazon Electronics dataset}
	%\label{sfigconncmp}
    \end{subfigure}
    \caption{HoloScope (HS) runs in near-linear time.}
    \label{figeffe}
\end{figure}

%\begin{figure}
%    \centering
%	\includegraphics[height=1.4in]{effeciencyexp}
%	\caption{Algorithm running time v.s.\ the number of edges.}
%	\label{figeffe}
%\end{figure}

Accurately detecting fraudulent blocks of lower density requires aggregating more sources of information. 
%Such a prior suspiciousness guarantees
%a better detection bound, and camouflage resistance.
%This shows that researchers gradually realized 
%that purely density detection is not robust for fraud detection
%problem anymore.
%However those prior weights do not adjust themselves, 
%given current selection of suspect users.
%From a synthetic example in Figure~\ref{fighyblock} (a),
%we can see that Fraudar shrinks the height of detection
%box (green box), compared with ``average degree'' (red box) 
%that without edge weighting. 
%But the prior weight stops Fraudar from being aware of 
%the posterior that many other users outside the green box
%also connect to the objects between column 500 and about 700 
%(green dash-dot line).
%And dense blocks of neutral ratings may not 
%be a suspicious block, as opposite to either very
%positive or negative ratings.
Consider the attribute of the creation time of edges: fraudulent attacks tend
to be concentrated in time, e.g., 
fraudsters may surge to retweet a message, 
creating one or more sudden bursts of activity, followed by sudden drops after the attack is complete. 
Sudden bursts and drops have not been
directly considered together in previous work.

Tensor based 
methods~\cite{jiang2015general,shin2016mzoom,shin2017dcube}
incorporate edge attributes into a multi-way tensor formulation,
e.g., IPs, rating scores and time.
However, those methods rely on time-binning to incorporate temporal information, and then treat time bins independently, which
loses information about bursts and drops, which our approach captures. 

%Fraudar~\cite{hooi2016fraudar} introduced a way 
%to consider extra information, by assigning a prior suspiciousness
%to nodes and edges. 
%However, the prior suspiciousness is predefined and fixed
%at the beginning which is not flexible.
%The prior suspiciousness is decided globally by some
%experiences, such as suspicious profiles of users or objects,
%and suspicious text of a review by a user for an object. 
%Besides, a sudden drop in Figure~\ref{figtsshowcase} (b)
%is another observations after fraudsters finishing 
%their attack, which are not considered in the previous work.

Therefore, we propose HoloScope, which combines suspicious 
signals from graph topology, temporal bursts and drops, 
and rating deviation. 
Our graph topology-based weighting scheme dynamically reweights objects according to our beliefs about which users are suspicious. 
% Our topology captures the contrast ``ratio'' of an object
% connected by suspicious users to the object's whole connections,
% in which the higher ratio implies that no other users
% are willing to connect to the suspicious object.
Temporally, HoloScope detects suspicious spikes of bursts and drops, which increases the time cost needed for fraudsters to conduct an attack.
%In the same way, HoloScope uses the contrast ratio
%of how many suspicious users create edges in a bursting period,
%and the sudden drop, to measure suspiciousness.
% The rating deviation of an object considers the differences of
% the rating scores between the suspicious users and 
% the others. 
In terms of rating deviation, our approach takes into account how much difference there is between an object's ratings given by suspicious users and non-suspicious users. 

In summary, our contributions are:
%Our Holoscope can improve the biased density block detection,
%and consider more properties like timing and rating scores, 
%so that the method can improve the precision, and increase
%the time cost and extra rate cost for fraudsters without being detected.
\begin{itemize}
    \item{\textbf{Unification of signals:}} we make holistic
	use of several signals, namely connectivity (i.e., topology), temporal bursts and drops,
	and rating deviation in a systematic way. 
	%Especially, HoloScope captures the spikes in the time
	%series of an object, i.e., sudden drop and burst, 
	%and the rate deviation of in-and-out suspicious users.
    \item{\textbf{Theoretical analysis of fraudsters' obstruction:}}
	we show that if the fraudsters use less than an upper bound
	of time for an attack, they will 
	cause a suspicious drop or burst. In other words, HoloScope obstructs fraudsters by increasing the time they need to perform an attack. 
    \item{\textbf{Effectiveness:}} we achieved higher accuracy 
	than the competitors on semi-real and real datasets. 
	In fact, HoloScope using only topology information (HS-$\alpha$)
	outperformed the graph-based competitors (see Fig.~\ref{sfigconncmp}),
	while HoloScope (HS) using all attributes
	achieved further improvement, and outperformed the tensor-based
	competitors (see Fig.~\ref{sfigallcmp} and \ref{sfigwbexp}).
    \item{\textbf{Scalability:}} 
 %    although HoloScope needs to dynamically
	% update the suspicious of objects, the algorithm can be as fast as
	% sub-quadratic time of the number of nodes, 
	% with a reasonable assumption. 
	HoloScope runs in subquadratic time in the number of nodes, under reasonable assumptions. 
	Fig.~\ref{figeffe} 
	shows that its running time increases near-linearly with the number of edges.
\end{itemize}
% The experimental results show that our HoloScope achieves
% higher accuracy in fraud detection by considering topology, 
% spikes patterns, and rate deviation simultaneously, 
% in the comparison with the-state-of-the-art methods. 
% Especially when the fruadulent density decreases,
% HoloScope can still keep a better performance
% at a lower density. And 
In addition, in Microblog, Sina Weibo\footnote{
    The largest Microblog service in China, \url{http://www.weibo.com}} 
data, HoloScope achieved higher F-measure than the competitors in detecting the
ground truth labels, with high precision and recall. 
The code of HoloScope is open-sourced for
reproducibility~\footnote{\url{https://github.com/shenghua-liu/HoloScope}}.
%HS is a friendly framework with the help of our $GreedyShaving$
%	and $GreedyInflating$ algorithms. HS can shaving any
%	initially detected blocks, such as those by Fraudar
%	and MZoom; or inflate a rank of suspicious users, such as
%	the spectral rank from SVD. 
	%the suspiciousness $P(v|A)$ introduced by HoloScope,
	%can also be embedded into the other metrics 
	%for fraud detection. 

%% file: 2-relwork.tex
\begin{table}[htbp]
  \centering
    \footnotesize
  \caption{Comparison between HoloScope and other fraud detection algorithms.}
    \begin{tabular}{L||cccccc||c}
	& \rotv{Fraudar~\cite{hooi2016fraudar}}& \rotv{SpokEn~\cite{prakash2009eigenspokes}}& 
	\rotv{CopyCatch~\cite{beutel2013copycatch}}& \rotv{CrossSpot~\cite{jiang2015general}}&   
	\rotv{BP-based methods~\cite{pandit2007netprobe,akoglu2013opinion}}&
	\rotv{M-Zoom/D-Cube~\cite{shin2016mzoom,shin2017dcube}}& \rotv{HoloScope}\\
    \hline
    %connectivity & 1     & 1     & 1     & 1     & 1     & 1     & 1 \\
    %quadruplet   &       &	 &	 & 1      & ?     & 1     & 1 \\
    scalability  & \checkmark & \checkmark & \checkmark 
	& \checkmark &\checkmark &\checkmark & \checkmark \\
    camouflage   &\checkmark   &    & ?       &\checkmark  &  &
	\checkmark & \checkmark\\
    hy-community &       &       & ?     & ?  &   &     & \checkmark \\
    spike-aware  &       &       & ?      &       &      &      & \checkmark \\
    %time density &       &       & \checkmark     & \checkmark  &   & \checkmark  & \checkmark \\
    %rate density &       &       &       & \checkmark &     & \checkmark     & \checkmark \\
    %divergent    &       &       &       & 1     &       &       & 1 \\
  %(multi-) burst  &       &       & ?     & ?    & ?    & ?     & 1 \\
    
    %divergent rate &       &       &       &       &       &       & 1 \\
    \end{tabular}%
  \label{tbrelworks}%
\end{table}%

\section{Related Works}

Most existing works study anomaly detection 
based on the density of blocks within 
adjacency matrices~\cite{prakash2009eigenspokes, jiang2014inferring},
or multi-way tensors~\cite{shin2016mzoom,shin2017dcube}. 
%OddBall~\cite{akoglu2010oddball} found new rules and patterns in 
%the distribution of eigenvalues for anomaly detection. 
%\cite{jiang2014inferring} detected suspicious patterns in the plots of 
%SVD (singular vector decomposition). 
%In stead of detecting density block by average degree~\cite{charikar2000greedy}, 
\cite{gibson2005discovering} and CoreScope~\cite{shin2016corescope}  
proposed to use Shingling and K-core algorithms respectively 
to detect anomalous dense block in huge graphs.
Taking into account the suspiciousness of each edge or node
in a real life graph potentially allows for more accurate detection. 
Fraudar~\cite{hooi2016fraudar} proposed to weight edges' suspiciousness 
by the inverse logarithm of objects' indegrees, to discount popular objects.
\cite{araujo2014beyond} found that the degrees in a large
community follow a power law distribution, forming hyperbolic structures. 
This suggests penalizing high degree objects to 
avoid unnecessarily detecting the dense core of hyperbolic community. 

In addition to topological density,
EdgeCentric~\cite{shah2015edgecentric} studied the distribution of rating 
scores to find the anomaly.
In terms of temporal attribute, 
the identification of burst period 
has been studied in~\cite{kumar2005bursty}.
A recent work, Sleep Beauty (SB)~\cite{ke2015defining} 
more intuitively defined the awakening time for a paper's citation burst for
burst period.
\cite{yang2011patterns} clustered the temporal patterns of 
text phrases and hash tags in Twitter.
%, and 
%\cite{myers2014bursty,paranjape2016motifs}
%studied the temporal dynamics of networks separately on ego-network and
%network motifs.
Meanwhile, \cite{gunnemann2014detecting,gunnemann2014robust}
modeled the time stamped rating scores with Bayesian model
and autoregression model respectively for anomalous behavior detection.
%Those works are mainly focus on the attributes
%for a single type of object, e.g., edges, ratings, and citations.
Even though \cite{xie2012review,li2016modeling,fei2013exploiting} 
have used bursty patterns to detect review spam, a sudden drop in
temporal spikes has not been considered yet. %is actually a more suspicious signal. 

However, aggregating suspiciousness signals from 
different attributes is challenging.
\cite{cormack2009reciprocal} proposed 
RRF (Reciprocal Rank Fusion) scores for combining different rank lists in information retrieval.
However, RRF applies to ranks, not suspiciousness scores.
%\begin{equation}
%    \nonumber
%    RRF(v|A) = \sum_{r \in R}\frac{1}{k + r(v|A)}
%\end{equation}
%where $r(\cdot)$ is the order in a rank, $R$ is
%the ranks obtained by all types of abnormality scores.

%Researchers used the tensor-based methods~\cite{jiang2015general,shin2016mzoom,shin2017dcube}
%to consider different attributes.
CrossSpot~\cite{jiang2015general}, a tensor-based algorithm,
estimated the suspiciousness of a block using a Poisson model. 
% However, the suspiciousness of every dimension is
% the KL divergence to a volume density of 
% the whole tensor, which does not consider 
% the differences between popular and unpopular objects.
However, it did not take into account the difference between popular and unpopular objects. 
Moreover, although CrossSpot, M-Zoom~\cite{shin2016mzoom} and
D-Cube~\cite{shin2017dcube}
can consider edge attributes like rating time and
scores via a multi-way tensor approach, 
they require a time-binning approach. When time is split into bins, 
attacks which create bursts and drops may not stand out clearly after
time-binning.
%treating each time bin as an independent dimension. 
The problem of choosing bin widths for histograms was
studied by Sturges~\cite{sturges1926choice} 
assuming an approximately normal distribution, and 
Freedman-Diaconis~\cite{freedman1981histogram}
based on statistical dispersion. However, the binning
approaches were proposed for the time series of a single
object, which is not clear for different kinds of objects 
in a real life graph, namely, popular products and
unpopular products should use different bin sizes.

Belief propagation~\cite{pandit2007netprobe} 
is another common approach for fraud detection
which can incorporate some specific 
edge attributes, like sentiment~\cite{akoglu2013opinion}.
%or use seeds discovered from suspicious patterns~\cite{jiang2014catchsync}
%as a prior. 
However, its robustness against adversaries which try to hide themselves is not well understood.

CopyCatch~\cite{beutel2013copycatch} detected lockstep behavior by 
maximizing the number of edges in blocks constrained within 
time windows. However, this approach ignores the distribution of edge creation times within the window, and does not capture bursts and drops directly. 
%CopyCatch also only supports a single time center, and hence cannot capture attacks which occur over multiple bursts. 
% Such a fixed-size time window for each object 
% can capture a dense time area, which does not 
% exactly focus on the involvement 
% from awakening to burst, and afterward
% drop. 
% Besides CopyCatch only supports one time center.
% The fraud attacks at a separate time in
% multiple groups (multi-bursts) may disturb the algorithm.
%Therefore, although it is not easy for fraudsters 
%to sufficiently scatter the firepower, they 
%may attack this detection by partitioning the 
%firepower into multiple groups, which results in
%multiple bursts.

Finally, we summarize our competitors compared to our HoloScope
in Table~\ref{tbrelworks}. We use ``hy-community'' to indicate 
whether the method can avoid detecting the naturally-formed
hyperbolic topology which is unnecessary for fraud detection.
HoloScope is the only one which 
considers temporal spikes (sudden bursts and drops and multiple bursts) 
and hyperbolic topology, in a unified suspiciousness framework.

%% file: 3-alg.tex
\section{Proposed Approach}

The definition of our problem is as follows.
\begin{problem}[informal definition]
    \textbf{Given} quadruplets ($user$, $object$, $timstamp$, $\#stars)$,
	    where $timestamp$ is the time that a $user$ rates an $object$, and
	    $\#stars$ is the categorical rating scores.
    \begin{itemize}
	\item[-]\textbf{Find} a group of suspicious users, and suspicious objects
	    or its rank list with suspiciousness scores,
	\item[-]\textbf{to optimize} the metric under the common 
	    knowledge of suspiciousness from topology,
	    rating time and scores.
    \end{itemize}
\end{problem}
To make the problem more general, 
$timestamps$ and $\#stars$ are optional. 
For example, in Twitter, we have  
$(user, object, timestamp)$ triples, where $user$ retweets 
a message $object$ at $timestamp$. 
In a static following network, 
we have pairs $(user, object)$, with
$user$ following $object$.

As discussed in previous sections, our metric should capture the following basic traits:
\begin{trait}[Engagement]
    \label{axiengage}
    Fraudsters engage as much firepower as
    possible to boost customers' objects, i.e., suspicious objects. 
    % with a certain concealing.
\end{trait}

\begin{trait}[Less Involvement]
    \label{axiinvolv}
    Suspicious objects seldom attract non-fraudulent users to connect with them.
\end{trait}

\begin{trait}[Spikes: Bursts and Drops]
    \label{axiburst}
    Fraudulent attacks are concentrated in time, sometimes over multiple waves of attacks, creating bursts of activity. Conversely, the end of an attack corresponds to sudden drops in activity.
    % The firepower of fraudsters is timing concentrate,
    % and sometimes multi-waved, causing burst points in
    % time series. After finishing fraud, a sudden drop is usually observed
    % due to objects' own fame.
\end{trait}

\begin{trait}[Rating Deviation]
    \label{axirate}
    % High-rating objects seldom attract other users than
    % defaming fraudsters to rate extremely low.
    % On the opposite, low-rating or neutral-rating objects seldom attract 
    % other users than fraudsters to rate extremely high.
    High-rating objects seldom attract extremely low ratings from normal users. Conversely, low-rating objects seldom attract extremely high ratings from normal users. 
\end{trait}

Thus we will show in the following sections, that our proposed metric 
can make holistic use of several signals, namely topology, temporal
spikes, and rating deviation, 
to locate suspicious users and objects satisfying the above traits.
That is the reason we name our method as HoloScope (HS).

\subsection{HoloScope metric}
To give a formal definition of our metric, we  
describe the quadruplets ($user$, $object$, $timstamp$, $\#stars$) 
as a bipartite and directed graph $\mathcal{G}=\{U,V,E\}$, 
which $U$ is the source node set, $V$ is the sink node set,
and connections $E$ is the directed edges from $U$ to $V$. 
Generally, graph $\mathcal G$ is a multigraph, i.e., multiple edges can be present between two nodes.
Multiple edges mean that a user can repeatedly comment or rate on the same product at a different time, as common in practice.
Users can also retweet message multiple times in the Microblog $Sina$ $Weibo$.
Each edge can be associated with rating scores ($\#stars$), 
and timestamps, for which the data structure is introduced
in Subsection~\ref{secburstdrop}.

%In the scenario of webpages/message sharing, source nodes $U$ are
%users, and sink nodes $V$ are webpages/messages. 
%When a user thumb up a webpage or retweet a message, 
%a directed edge is formed between them. 
%In the user following graph, source nodes $U$ are the fans, and 
%sink nodes $V$ are those being followed. $U$ and $V$ may 
%contain the same users.
%As for rating a product, 
%edges $E$ is the rating relations between users $U$ and
%objects $V$. Each edge can be associated with rating scores 
%(#stars), and timestamps.

Our HoloScope metric detects fraud from three perspectives:
topology connection, timestamps, and rating score.
To easily understand the framework, we first introduce
the HoloScope in a perspective of topology connection. 
Afterwards, we show how we aggregate the other 
two perspectives into the HoloScope.
%, with the help 
%of the contrast suspiciousness, and 
%the global suspiciousness inherited from Fraudar.
We first view $\mathcal G$ as 
a weighted adjacency matrix $\mathbb{M}$, with the number of 
multiple edges (i.e., edge frequency) as 
matrix elements.

%We firstly consider popularizing products, only keeping edges
%of high rating scores. 
%For the fraud detection of general product rating, 
%we will discussed in later sections.
%So before introducing any properties on multiple
%edges, we can simply treat graph $\mathcal G$ as a frequency matrix.
%For detecting defaming fraud,
%edges of low rating scores are kept, and the algorithm can
%be easily extended to such a case, which will be shown 
%in later sections.

Our goal is to find lockstep behavior 
of a group of suspicious source nodes $ A\subset U$ 
who act on a group of sink nodes $B\subset V$.
Based on Trait~\ref{axiengage}, the total
engagement of source nodes $A$ to sink nodes $B$ can be
basically measured via density measures. 
There are many density measures,
such as arithmetic and
%~\cite{goldberg1984finding,gallo1989fast},
geometric average degree. 
%~\cite{charikar2000greedy}.
Our HoloScope metric allows for any such measure.
However, as the average degree metrics have a bias toward including too many nodes, we use a
measure denoted by $D(A,B)$ as the basis of the HoloScope, defined as:
% Later we will show that HoloScope can offer a guarantee 
% of a precise location with a dynamic weight.
% Later we will show that HoloScope achieves greater 
% accuracy using a dynamic weighting scheme. 
\begin{equation}
    \label{eqD}
    D(A,B) =\frac{\sum_{v_i \in B} f_A(v_i)}{|A|+|B|}
\end{equation}
where $f_A(v_i)$ is the total edge frequency 
from source nodes $A$ to a sink node $v_i$. 
$f_A(v_i)$ can also be viewed as an engagement from $A$ to $v_i$, 
or $A$'s lockstep on $v_i$, which is defined as
\begin{equation}
    \label{eqfav}
    f_A(v_i) = \sum_{(u_j,v_i) \in E \land u_j \in A}{\sigma_{ji}\cdot e_{ji}}
\end{equation}
where constant $\sigma{ji}$ are the global suspiciousness on 
an edge, and $e_{ji}$ is the element of adjacency matrix $\mathbb{M}$,
i.e., the edge frequency between a node pair $(u_j,v_i)$.
The edge frequency $e_{ji}$ becomes a binary in a simple graph.
The global suspiciousness as a prior can come from the degree, 
and the extra knowledge on fraudsters, such as 
duplicated review sentences and unusual behaving time.
%bryan edit
%where constants $e_{ji}$ and $\sigma_{ji}$ are the frequency
%and the global suspiciousness of the multiple edges $(u_j,v_i)$ respectively. 
%$\Omega$ is the average global suspiciousness 
%of the nodes in $A$ and $B$:
%\begin{equation}
%\Omega = \frac{\sum_{i\in A\cup B}\gamma_i}{|A|+|B|}
%\end{equation}
%where $\gamma_i$ is the global 
%suspiciousness of source or sink node $i$.
%The global suspiciousness of a node is a prior which can come from its degree, 
%or any extra knowledge about fraudsters, such as 
%duplicated review sentences and unusual temporal behavior.

To maximize $D(A,B)$, the suspicious source nodes $A$ and 
the suspicious sink nodes $B$ are mutually dependent.
% ; conditioning on one group, it brings more information about the other group. 
Therefore, we introduce \emph{contrast suspiciousness} 
in an informal definition: 
\begin{definition}[contrast suspiciousness]
    The contrast suspiciousness denoted as $P(v_i \in B|A)$ 
    is defined as the conditional likelihood of a sink node $v_i$ 
    that belongs to $B$ (the suspicious object set),
    given the suspicious source nodes $A$.
\end{definition}
A visualization of the contrast suspiciousness 
is given in Fig.~\ref{figdef}.
The intuitive idea behind contrast suspiciousness is that
in the most case, we need to judge the suspiciousness of 
objects by currently chosen suspicious users $A$, e.g.,
an object is more suspicious if very few users not in $A$ 
are connected to it (see Trait~\ref{axiinvolv}); the sudden burst of 
an object is mainly caused by $A$ (see Trait~\ref{axiburst}); 
or the rating scores from $A$ to an object are quite different 
from other users (see Trait~\ref{axirate}).
Therefore, such suspiciousness makes use of the contrasts 
between users in $A$ and users not in $A$ or the whole set. 

\begin{figure}
\centering
   {\includegraphics[width=0.24\textwidth]{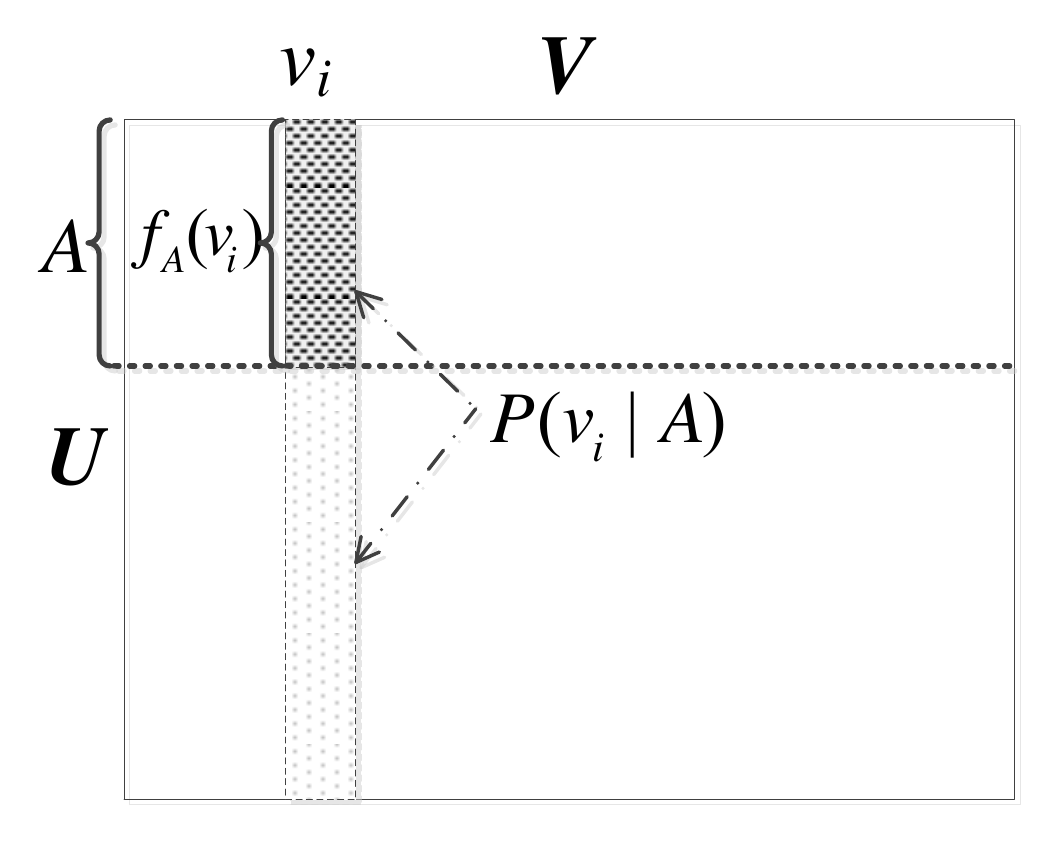}}
    \caption{An intuitive view of our definitions in the HoloScope.}
\label{figdef}
\end{figure}

Finally, instead of maximizing $D(A,B)$, 
we maximize the following expectation of 
suspiciousness $D(A,B)$ over the probabilities $P(v_i \in B|A)$:
\begin{align}
    \label{eqobj}
    \nonumber
    \max_{A} HS(A) &:= \mathbb{E}\left[D(A,B)\right] \\
	     &= \frac{\displaystyle 1}{\displaystyle |A| + 
	     \sum_{k\in V}P(v_k|A)} {\displaystyle\sum_{i\in V}
	     f_A(v_i)P(v_i|A)} 
	     %+ \mathbb{E}[\Omega]
\end{align}
where for simplicity we write $P(v_i|A)$ to mean $P(v_i \in B|A)$.
% as the contrast suspiciousness
% to represent the conditional probability that a
% sink node $v_i$ is suspicious, i.e., $v_i\in B$. 
$1-P(v_i|A)$ is the probability of $v_i$ being a normal sink node.
We dynamically calculate the contrast suspiciousness for all the objects, 
after every choice of source nodes $A$. 
%The expectation of $\Omega$ is then
%\begin{equation}
%    \nonumber
%    \mathbb{E}[\Omega] = \frac{\sum_{j\in A}\gamma_j + 
%	     \sum_{j\in V}\gamma_j P(v_j|A)}{|A|+\sum_{j\in V}P(v_j|A)}.
%\end{equation}
%Thus, if all the global suspiciousness 
%of source nodes and sink nodes are equal, i.e., 
%$\gamma_i = \gamma_j=\gamma \ \forall \ i,j$, then the
%expectation will be a constant $\gamma$,
%which can be ignored.

Using this overall framework for our proposed metric $HS(A)$, 
we next show how to satisfy the remaining Traits. 
To do this, we define contrast suspiciousness 
$P(v_i|A)$ in a way that takes into account various edge attributes. 
This will allow greater accuracy particularly for detecting low-density blocks.
%And the way to aggregate multiple properties.

\subsubsection{HS-$\alpha$: Less involvement from others.}

Based on Trait~\ref{axiinvolv}, a sink node should be more suspicious 
if it only attracts connections from the suspicious source nodes $A$, and less 
from other nodes. 
Mathematically, we capture this by defining
\begin{equation}
    P(v_i | A) \propto q(\alpha_i), \text{~where~} 
    \alpha_i=\frac{f_A(v_i)}{f_{U}(v_i)} 
\end{equation}
where $f_U(v_i)$ is the weighted indegree of sink node 
$v_i$. Similar to $f_A(v_i)$, the edges are weighted by global
suspiciousness.
$\alpha_i$ measures the involvement ratio of $A$ in the activity of sink node $v_i$.
The scaling function $q(\cdot)$ is our belief about how this ratio relates to suspiciousness, and we choose the exponential form 
$q(x)=b^{x-1}$, where base $b>1$. 
% In this way, 
% the contrast suspiciousness from the connectivity involvement,
% and other properties like timing and rating, can be aggregated
% together by a product as a joint probability.
% So the exponential function $q(x)$ can substantially sum
% the suspiciousness scores in the exponent. 
% In such a way, scaling function can prevent the joint probability
% being counteracted by a very small probability 
% from other perspectives.
%The choice of base $b$ is not very sensitive,
%and we will show some advice later.

\begin{figure}
\centering
    \begin{subfigure}[b]{0.24\textwidth}
	\centering
	{\includegraphics[width=\textwidth]{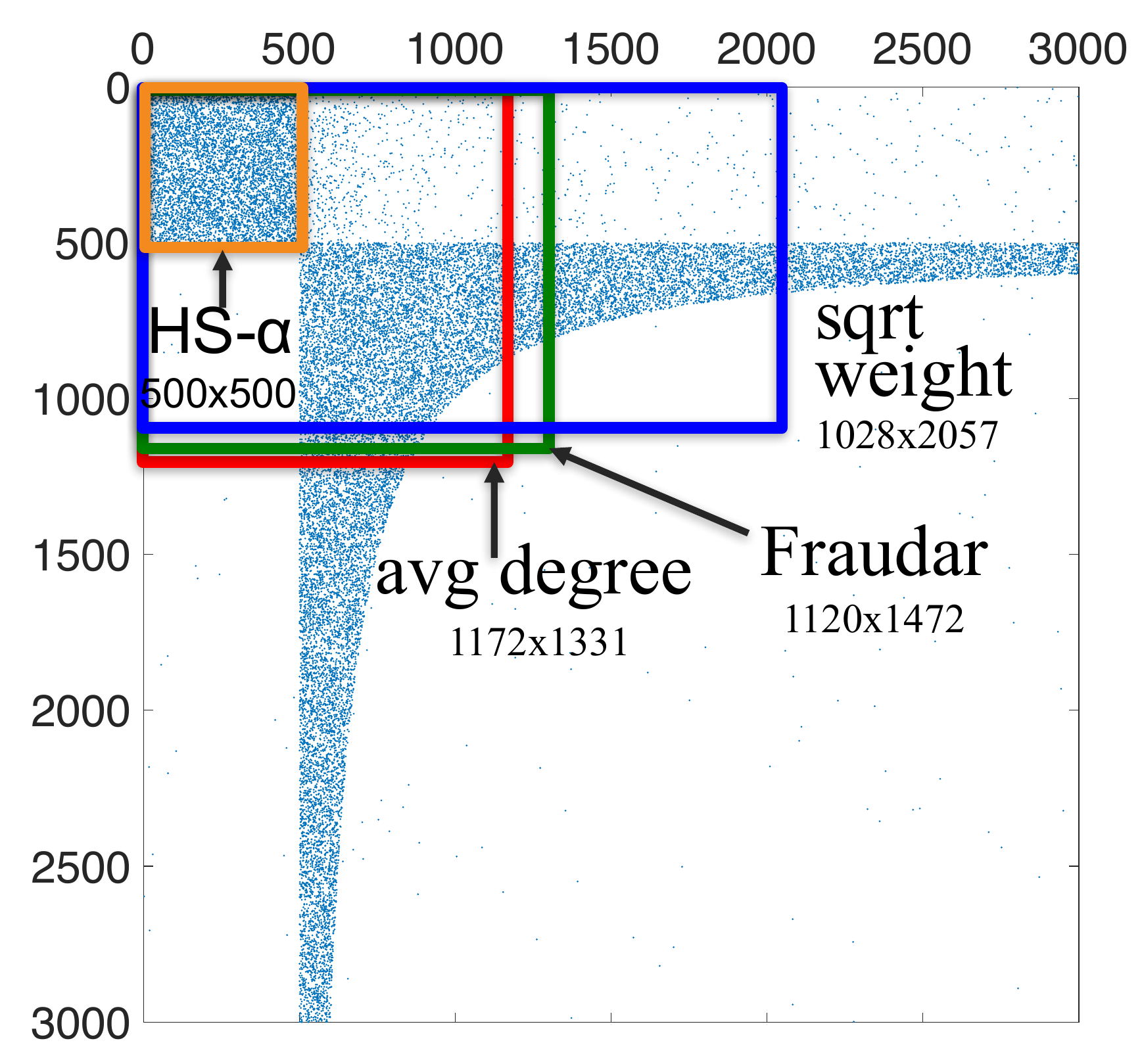}}
	\caption{HS-$\alpha$ finds the exact dense block in the synthetic data}
	\label{sfighyper}
    \end{subfigure}
    \begin{subfigure}[b]{0.22\textwidth}
	\centering
	{\includegraphics[width=\textwidth]{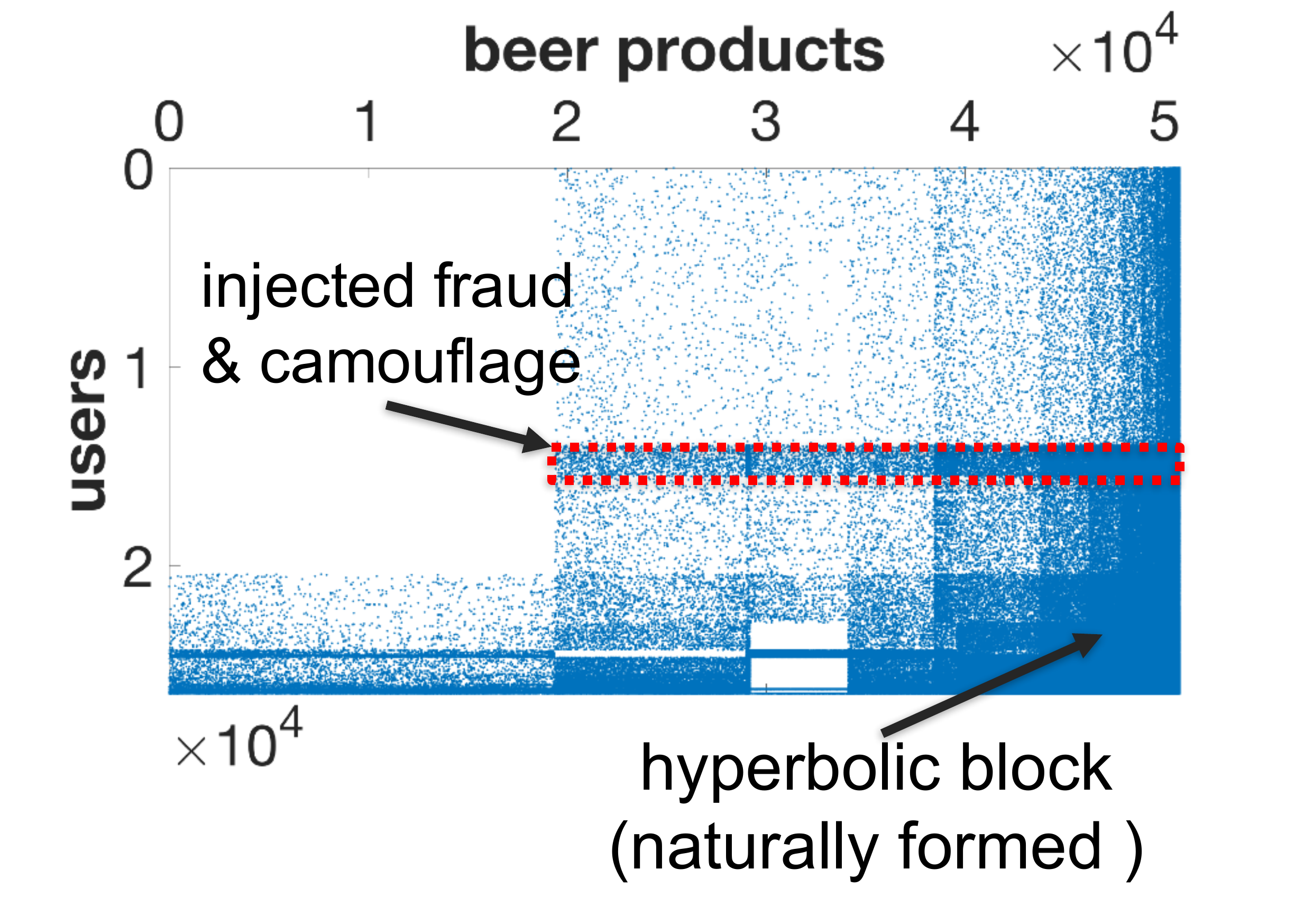}}
	\caption{Hyperbolic community in BeerAdvocate data}
	\label{sfigbeerhyper}
    \end{subfigure}
    \caption{(a) The synthetic data consists of hyperbolic 
    and rectangular blocks, with volume density around 0.84 and 0.60
    respectively. The camouflage is randomly biased to columns of high degree.
    %Existing methods without considering contrast
    %suspiciousness, i.e.\ Trait~\ref{axiinvolv}, do not detect
    %the most suspicious area precisely. They sometimes even miss 
    %smaller density rectangular areas, due to the existence 
    %of hyperbolic connections.
    (b) A real data of naturally-formed hyperbolic community, and injected
    dense block. The injection is $2000\times 200$ with biased camouflage.}
    \label{fighyblock}
\end{figure}

%As previous work showed,
%large communities form hyperbolic 
%structures, which is generated in 
%our synthetic data (see the lower-right block 
%in Fig.~\ref{sfighyper}), and 
%also exists in real BeerAdvocate data (see Fig.~\ref{sfigbeerhyper}).
%Thus comparing the rectangular dense block and 
%the hyperbolic core (i.e., the upper-left part of hyperbolic block in Fig.~\ref{sfighyper}), 
%which one is more suspicious?
%By examination in the scenario of online reviews,
%the products in the hyperbolic core are also rated by many other
%people not in the core, with high scores;
%In the case of rectangular block, the products seldom
%attract other people to give high scores, which may not be
%really a good products.
%So the rectangular dense block is more suspicious.

As previous work showed,
large communities form hyperbolic 
structures, which is generated in 
our synthetic data (see the lower-right block 
in Fig.~\ref{sfighyper}), and 
also exists in real BeerAdvocate data (see Fig.~\ref{sfigbeerhyper}).
For clarity, our HoloScope method are denoted as HS-$\alpha$
when it is only applied on a connection graph.
The results of the synthetic data show that
HS-$\alpha$ detected the exact dense 
rectangular block ($b=128$),
while the other competitors included
a lot of non-suspicious nodes from the core part of 
hyperbolic community resulting in low accuracy.
In the beer review data from the BeerAdvocate website, 
testing on different fraudulent density (see Fig.~\ref{sfigconncmp}),
our HS-$\alpha$ remained at high accuracy, while the
other methods' accuracy drops quickly 
when the density drops below 70\%. 

The main idea is that HS-$\alpha$ can do better because
it dynamically adjusts the weights for sink nodes, 
penalizing those sink nodes that also have many 
connections from other source nodes not in $A$.
% The weights are used for every edge connected 
% to those sink nodes.
%the methods including Fraudar and those based on average degree
%or square-root (sqrt) weight, output a solution
%overlapping both the rectangular block in the top-left and the hyperbolic block, 
%to maximize their objectives.
In contrast, 
although Fraudar proposed to
penalize popular sink nodes based their indegrees, 
these penalties also scaled down the weights of suspicious edges.
The Fraudar (green box) only improved the 
unweighted ``average degree'' method (red box) by a very limited
amount. Moreover, with a heavier
penalty, the ``sqrt weight'' method (blue box)
achieved better accuracy on source nodes but 
worse accuracy on sink nodes, since those methods used globally fixed weights, 
and the weights of suspicious were penalized as well.
Hence the hyperbolic structure pushes those methods to include more nodes
from its core part.
%Since those methods use globally fixed weights, and the suspicious and
%non-suspicious nodes are penalized in the same way, 
%the methods then need to include
%so that the algorithm need to include more nodes to maximize the density.

In summary, our HS-$\alpha$ using dynamic contrast suspiciousness
can improve the accuracy of fraud detection in `noisy' 
graphs (containing hyperbolic communities), 
even with low fraudulent density.

\begin{figure*}
\centering
    \begin{subfigure}[t]{0.34\textwidth}
	\centering
	\includegraphics[height=1.75in]{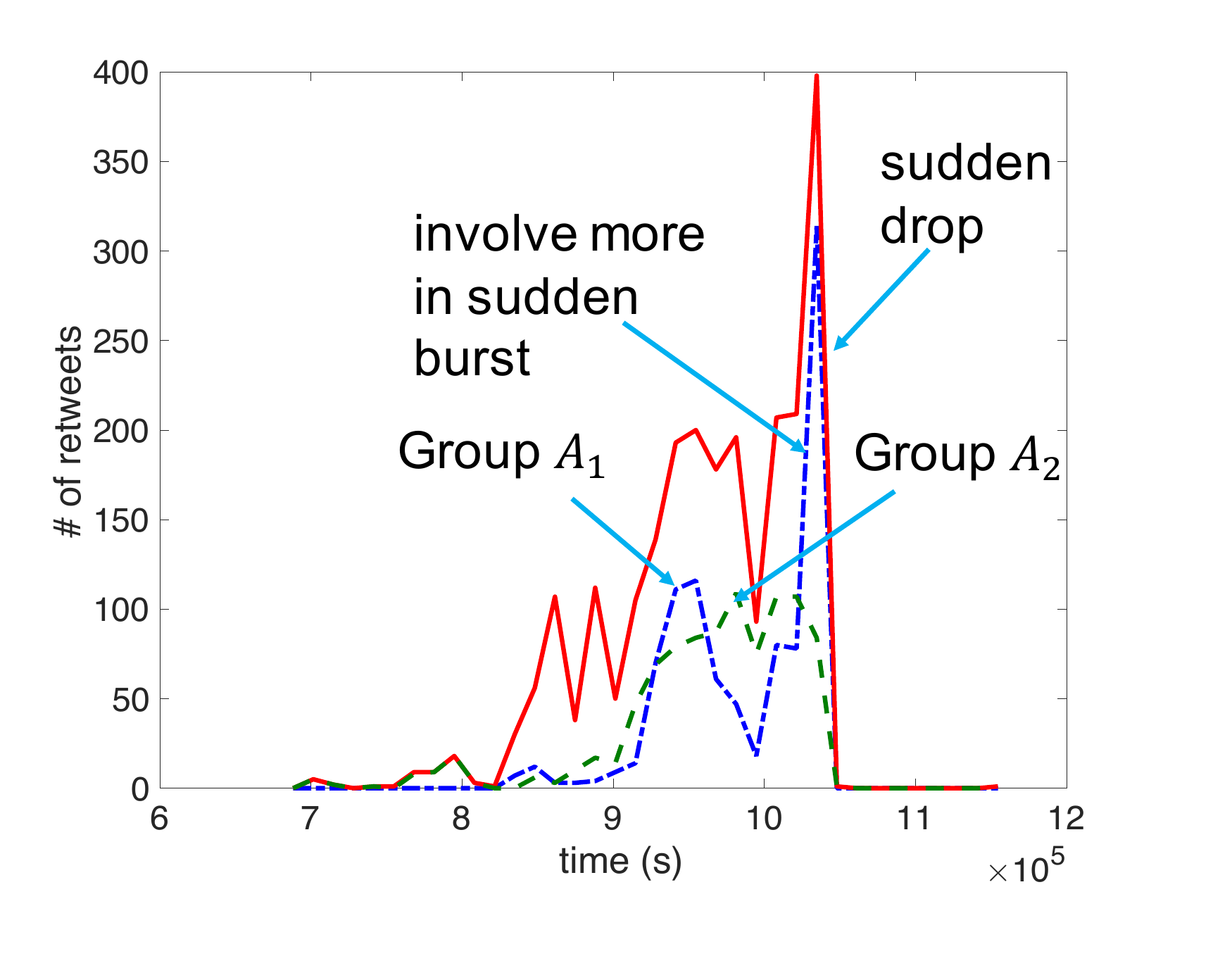}
	\caption{Group $A_1$ is 
	more suspicious than $A_2$, due to the sudden burst and drop}
	\label{sfigwbts}
    \end{subfigure}
    \begin{subfigure}[t]{0.31\textwidth}
	\centering
	\includegraphics[height=1.7in]{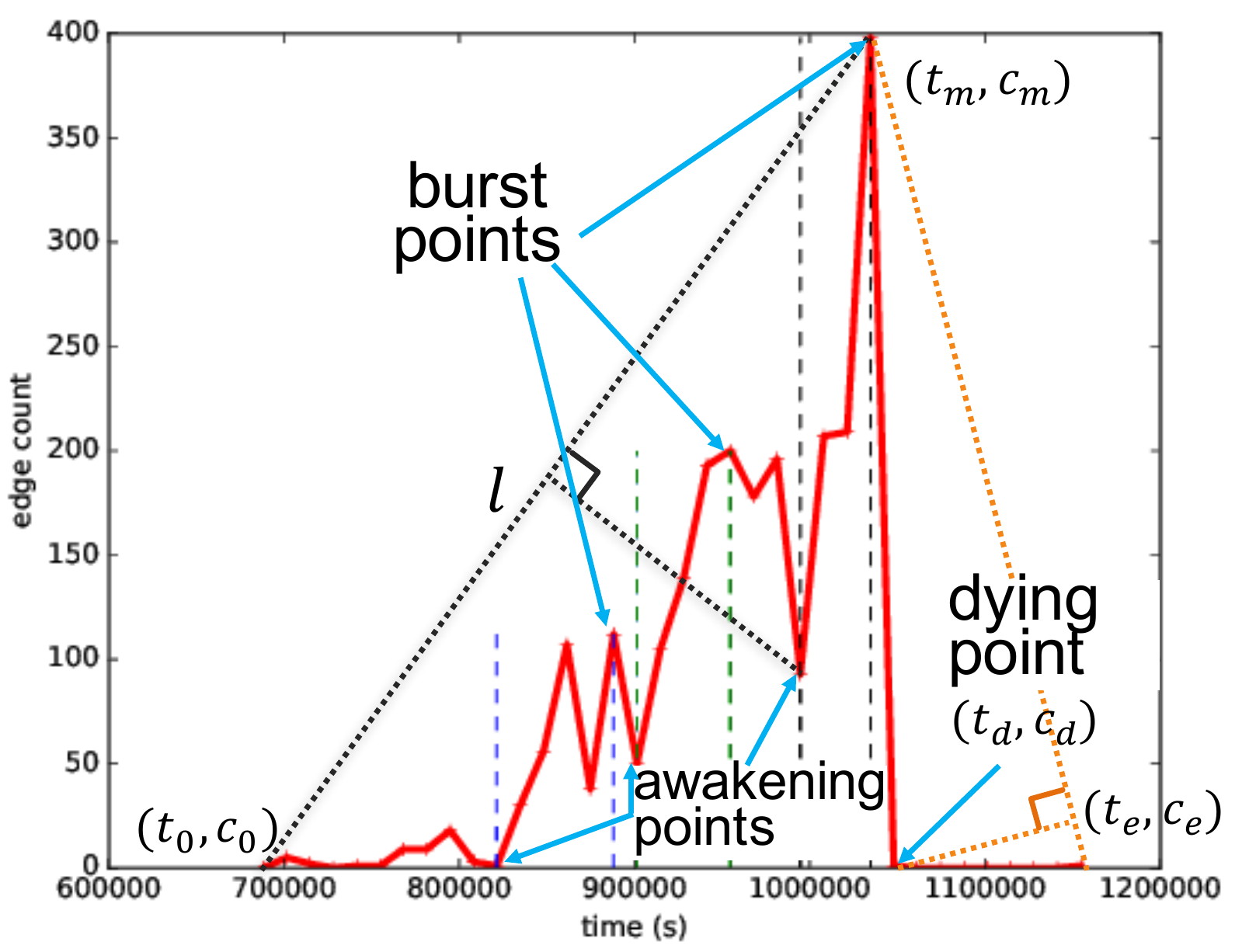}
	\caption{Detection of temporal bursts and drops}
	\label{sfigwbbd}
    \end{subfigure}
    \begin{subfigure}[t]{0.3\textwidth}
	\centering
	{\includegraphics[height=1.0in]{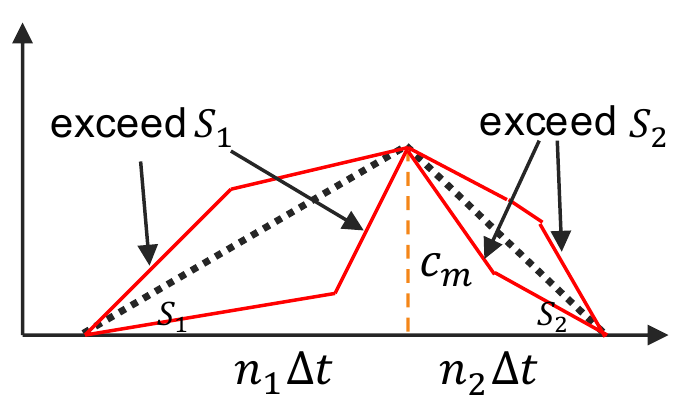}}
	\caption{Proof the time cost obstruction}
	\label{figtimecost}
    \end{subfigure}
    \caption{(a) and (b) are the time series (histogram) of a real message being
    retweeted in Microblog, Sina Weibo. The horizontal axis is the seconds 
    after $2013-11-1$. (c) illustrates our proof of time cost obstruction.
    }
\label{figtsshowcase}
\end{figure*}

\subsubsection{Temporal bursts and drops}
\label{secburstdrop}
%define time series $\mathcal T$
Timestamps for edge creation are commonly available in most real settings.
% With the extra property of timestamps, 
% $\mathcal{G}$ can be traditionally represented as 
% a multi-way tensors. 
% However, the tensor could be very sparse,  
% and representing the time in a way of tensor is not a 
% trivial thing, since time is never ending.
If two subgroups of Microblog users have the same number of retweets to a message, 
can we say they have the same suspiciousness?
As an example shown in Fig.~\ref{sfigwbts},
the red line is the time series (histogram of time bins) of 
the total retweets of a message in Microblog, Sina Weibo. 
The blue dotted line and green dashed line are
the retweeting time series respectively from 
user groups $A_1$ and $A_2$.
The two series have  
the same area under the time series curves, i.e., the same number of retweets.
However, considering that fraudsters tend to surge
to retweet a message to reduce the time cost, 
the surge should create one or more sudden bursts, 
along with sudden drops. 
Therefore, the suspiciousness of user groups $A_1$ and $A_2$ 
become quite different even though they have
the same number of retweets, which cannot be detected
solely based on connections in the graph.
Thus we include the temporal attribute into our HoloScope
framework for defining contrast suspiciousness.

%To include the temporal attribute, 
%we introduce a simple but effective data structure, 
%the \emph{property indexing matrix} (PIM).
%For each kind of
%property, we store them in a list of property entries. 
%Each property entry is associated with a pair of source node and sink node.
%The timestamps of multiple edges form a list and are associated
%with the corresponding entry.
%The PIM matrix has a similar structure to the weighted adjacency matrix of $\mathcal G$, except that it stores the index of a property entry 
%as an element, instead of the connection frequency.
%With the PIM matrix, we can easily 
%find the timestamps of edges from source nodes in $A$ via
%``slicing'' operations on the PIM.

Denote the list of
timestamps of edges connected to a sink node $v$ as $T_v$. 
To simplify notation, we use $T$ without subscript when talking
about a single given sink node $v$. Let $\mathcal T$=\{($t_0$, $c_0$), ($t_1$, $c_1$),
$\cdots$, ($t_e$, $c_e$)\} as the \emph{time series} of $T$, i.e., the histogram of $T$. The count $c_i$ is the 
number of timestamps in the time bin 
$[t_i-\Delta t/2, t_i+\Delta t/2)$, with 
bin size $\Delta t$.
The bin size of histogram is calculated according to
the maximum of Sturges criteria and the
robust Freedman-Diaconis' criteria as mentioned in related works. 
It is worth noticing that the HoloScope can \emph{tune} different 
bin sizes for different sink nodes, e.g., popular objects need
fine-grained bins to explore detailed patterns.
Hence, the HoloScope is more flexible than tensor 
based methods, which use a globally fixed bin size. 
Moreover, the HoloScope can update the time series at
a low cost when $T$ is increasing.

To consider the burst and drop patterns described in Trait~\ref{axiburst}, 
we need to decide the start point of a burst and the end point of
a drop in time series $\mathcal T$. Let the burst point be ($t_m$, $c_m$), 
having the maximum value $c_m$.
According to the definition in previous work ``Sleeping Beauty'', 
we use an auxiliary straight line from the beginning to the burst point 
to decide the start point, named
the \emph{awakening point} of the burst.
Fig.~\ref{sfigwbbd} shows the time series $\mathcal T$ (red polygonal line)
of a message from Sina Weibo, the auxiliary straight line $l$
(black dotted line) from the lower left point ($t_0$, $c_0$) to upper right
point ($t_m$, $c_m$), and the awakening point for the maximum point
($t_m$, $c_m$), which is defined as the point along 
%After drawing a straight line $l$, we only consider 
%the awakening point under line $l$.
% Thus the awakening point 
%in the paper is then defined as
%the point which 
the time series $\mathcal T$ which 
%is under the line $l$, and 
maximizes the distance to $l$. 
As the dotted line perpendicular to $l$ suggests in this
figure, the awakening point $(t_a, c_a)$ satisfies 
\begin{equation}
    \label{eqta}
    t_a = \argmax_{(c,t)\in \mathcal T, t<t_m}{\frac{|(c_m-c_0)t-(t_m-t_0)c +
    t_mc_0 - c_mt_0|}
    {\sqrt{{(c_m-c_0)}^2 + {(t_m-t_0)}^2}}}
\end{equation}

Finding the awakening point for one burst is not enough, as multiple bursts may be present. Thus, sub-burst points and the associated awakening points 
should be considered. We then propose a recursive algorithm
$MultiBurst$ in Alg.~\ref{algmultiburst} for such a purpose.

\begin{algorithm}[htbp]
\caption{$MultiBurst$ algorithm.}
\small
\label{algmultiburst}
\begin{algorithmic}
    \STATE{\bf Input} Time series $\mathcal T$ of sink node $v$, beginning index
    $i$, end index $j$\\
    \STATE{\bf Output} A list of awakening-burst point pairs,\\
        \qquad $s_{am}$: slope of the line passing through each point pair,\\
	\qquad $\Delta c$: altitude difference of each point pair.
    \STATE\textbf{If} $j-i<2$ \textbf{then return}
    \STATE $(t_m, c_m)$ = point of maximum altitude between indices $i$ and $j$.
    \STATE $(t_a, c_a)$ = the awakening point as Eq~(\ref{eqta}) between
    indices $i$ and $j$.
    \STATE $\Delta c_{am} = c_m - c_a$, and  $s_{am}=\rfrac{\Delta c_{am}}{(t_m - t_a)}$
    \STATE Append \{$(t_a, c_a), (t_m, c_m)$\}, $s_{am}$, and $\Delta c_{am}$ into
    the output.
    \STATE $MultiBurst(\mathcal T, i, a-1)$
    \STATE $k$ = Find the first local min position from indices $m+1$ to $j$ 
    \STATE $MultiBurst(\mathcal T, k, j)$
\end{algorithmic}
\end{algorithm}

After finding awakening and burst points, 
the contrast suspiciousness of burst awareness
satisfies $P(v_i | A) \propto q(\varphi_i)$, where $\varphi_i$
is the involvement ratio of source nodes in $A$ in multiple bursts.
Let the collection of timestamps from $A$ to sink node $v_i$ be 
$T_A$. Then,
\begin{equation}
    \label{eqbeta}
    \varphi_i = \frac{\Phi(T_A)}{\Phi(T_U)} \text{, and }
    \Phi(T) = \displaystyle\sum_{(t_a,t_m)}\Delta c_{am}\cdot s_{am}
    \sum_{t\in T}\mathbf{1}(t\in [t_a, t_m])
    %\frac{\displaystyle\sum_{(t_a,t_m)}\Delta c_{am}\cdot s_{am}
    %\sum_{t\in T_A}\mathbf{1}(t\in [t_a, t_m])}{ 
    %\displaystyle\sum_{(t_a,t_m)}\Delta c_{am}\cdot s_{am}
    %\sum_{t\in T_U}\mathbf{1}(t\in [t_a, t_m])}
\end{equation}
where $s_{am}$ is the slope from the output of $MultiBurst$ algorithm.
Here $s_{am}$ is used as a weight based on how steep the current burst is. 
This definition of suspiciousness satisfies Trait~\ref{axiburst}.
It is worth noticing that the $MultiBurst$ algorithm only needs to 
be executed once. With the preprocessed awakening and burst points,
the contrast suspiciousness of edges connected to $v$ has $O(d_v)$
complexity, where $d_v$ is the degree of sink node $v$. 
Hence the complexity for overall sink nodes are $O(|E|)$.

In fact, sudden drops are also a prominent pattern
of fraudulent behavior as described in Trait~\ref{axiburst},
since after creating the attack is complete,
fraudsters usually stop their activity sharply.
%Fig.~\ref{sfigwbdrop} shows the heat map of sink nodes'
%indegrees and the drop slopes weighted by drop height, in which
%a large weighted slope implies a sudden drop.
%The data of the heat map consists of a user retweeting a message
%at a timestamp from Microblog system, Sina Weibo.
%In the upper right of the heat map, 
%we can see that the messages with large weighted drop slopes have been removed from the 
%system, implying suspiciousness.
%Thus a sudden drop measured by weighted drop slope
%can be used as a global suspiciousness measure for each edge
%connected to the sink node in the HoloScope metric $HS(A)$.
To make use of the suspicious pattern of a sudden drop,
we define the $dying$ $point$ as the end of a drop. 
As Fig.~\ref{sfigwbbd} suggests, 
another auxiliary straight line is drawn from the 
highest point ($t_m$, $c_m$) to the last point
($t_e$, $c_e$). The dying point ($t_d$, $c_d$) can be found by
maximizing the distance to this straight line.
Thus we can discover the ``sudden drop'' by the absolute slope value 
$s_{bd}$=$\displaystyle\sfrac{(c_m-c_d)}{(t_d - t_m)}$ between 
the burst point and the dying point. 
Since there may be several drops in a fluctuated time
series $\mathcal T$, we choose the drop with the maximum fall.
To find the maximum fall, we also need a recursive algorithm,
similar to Alg.~\ref{algmultiburst}:

\textit{ 1) Find a 
maximum point ($t_m$, $c_m$), and the corresponding 
dying point ($t_d$, $c_d$) by definition; 2) Calculate
the current drop slope $s_{bd}$, and the drop fall $\Delta c_{bd}$
= $c_m - c_d$; 3) Recursively find drop slope and drop fall
for the left and right parts of $\mathcal T$, i.e., $t<t_m$
and $t\geq t_d$ respectively.}

As a result, the algorithm returns the maximum drop fall 
$\Delta c_{bd}$, and its drop slope $s_{bd}$, which it has
found recursively.
Finally, we use the weighted drop slope $\Delta c_{bd} \cdot s_{bd}$
as a global suspiciousness in equation (\ref{eqfav}),
to measure the drop suspiciousness.
Each edge connected to the sink node $v$ is assigned 
the same drop suspiciousness.
We use a logarithm scale for smoothing those edge weights.

With this approach to detect bursts and drops, 
we now show that this provides a provable time obstruction for fraudsters. 
\begin{theorem}
    Let $N$ be the number of edges that fraudsters want 
    to create for an object.
    If the fraudsters use time less than 
    $\tau\geq\sqrt{\frac{2N\Delta t(S_1+S_1)}{S_1\cdot S_1}}$,
    then they will be tracked by a suspicious burst or drop,
    where $\Delta t$ is the size of time bins, and 
    $S_1$ and $S_2$ are the slopes of normal rise and
    decline respectively. 
\end{theorem}
\begin{proof}
    The most efficient way to create $N$ edges is to have one
    burst and one drop, otherwise more time is needed.
    As shown in Fig.~\ref{figtimecost}, 
    in order to minimize the slope,
    every point in the time series should in line with the two 
    auxiliary straight lines to the highest point $c_m$, 
    separately from the awakening and dying points.
    Otherwise, the slopes will exceed the normal values 
    $S_1$ and $S_2$. 
    Hence we only consider the triangle with the auxiliary lines as its
    two edges.
    It is worth noticing that a trapezoid
    whose legs have the same slopes as the triangle's edges cannot have
    a shorter time cost.
Then 
    \begin{equation}
	\nonumber
	\frac{c_m}{n_1\Delta t}  =  S_1\text{,\qquad }
	\frac{c_m}{n_2\Delta t}  =  S_2\text{,\qquad } 
	(n_1+n_2) \cdot c_m  =  2N'.
    \end{equation}
Here $n_1$ and $n_2$ are the number of time bins before and after the burst.
$N'$ is the total number of rating edges, and $N'\geq N$ consider
some edges from normal users.
Thus, solving the above equations, we have
    \begin{equation}
	\nonumber
	\tau = (n_1+n_2)\Delta t = \sqrt{\frac{2N'\Delta t(S_1+S_2)}{S_1\cdot S_2}}
	\geq \sqrt{\frac{2N\Delta t(S_1+S_2)}{S_1\cdot S_2}}
    \end{equation}
\end{proof}
We also have the height of burst,
	$c_m$ $\geq$ $\sqrt{\frac{2N\Delta t S_1S_2}{S_1 + S_2}}$.
Thus, the maximum height of time series $\mathcal T$ cannot 
be larger by far than that of a normal sink node. That is the reason that we
use the weighted $\varphi_i$ in equation (\ref{eqbeta}) and weighted
drop slope in equation (\ref{eqfav}).

\subsubsection{Rating deviation and aggregation}
We now consider edges with categorical attributes such as
rating scores, text contents, etc.
For each sink node $v_i$, 
we use the KL-divergence $\kappa_i$ between the distributions
separately from the suspicious source nodes $A$ and 
the other nodes, i.e., $U \setminus A$.
We use $U\setminus A$ for KL-divergence 
instead of the whole source nodes $U$, 
in order to avoid the trivial
case where most of the rating scores are from $A$.
The rating deviation $\kappa_i$ is scaled into
$[0,1]$ by the maximum value 
before being passed into function 
$q(\cdot)$ to compute contrast suspiciousness.
The neutral scores can be ignored in the KL-divergence
for the purpose of detecting fraudulent boosting or defamation.
Moreover, rating deviation is meaningful 
when both $A$ and $U\setminus A$ 
have the comparable numbers of ratings.
Thus, we weighted $\kappa_i$ by a balance factor,
$\min\displaystyle\{$ $\displaystyle\sfrac{f_{A}(v_i)}{f_{U\setminus A}(v_i)}$, 
$\displaystyle\sfrac{f_{U\setminus A}(v_i)}{f_{A}(v_i)}$ $\displaystyle\}$.
 
To make holistic use of different signals, i.e.,
topology, temporal spikes, and rating deviation,
we need a way to aggregate those signals together.
We have tried to use RRF (Reciprocal Rank Fusion) scores 
from Information Retrieval, and wrapped the scores
with and without scaling function $q(x)$.
Compared to RRF score, we found 
that a natural way of joint
probability by multiplying those signals together: 
\begin{equation}
    \label{eqfinalp}
P(v_i|A)=b^{\alpha_i+\varphi_i+\kappa_i-3},
\end{equation}
was the most effective way to aggregate. In a joint 
probability, we can consider the absolute suspicious 
value of each signal, as opposite to the only use
of ranking order. Moreover, being wrapped with $q(x)$,
the signal values cannot be canceled out by multiplying
a very small value from other signals. A concrete example
is that a suspicious spike can still keep 
a high suspiciousness score by multiplying a 
very small score from low fraudulent density.

Moreover, HoloScope dynamically updates the contrast suspiciousness
$P(v_i|A)$. Thus the sink nodes being added with camouflage
will have a very low contrast suspiciousness, with respect to the 
suspicious source nodes $A$. This offers HoloScope the resistance to camouflage.

\subsection{Algorithm}

Before designing the full algorithm for large scale datasets,
we firstly introduce the most important sub-procedure $GreedyShaving$ in Alg.~\ref{algshaving}.
%and inflating procedure Alg~\ref{alginflate}.

At the beginning, this greedy shaving procedure 
starts with an initial set $A_0\subset U$ as input.
It then greedily deletes source nodes from $A$, according
to users' scores $\mathcal S$:
\begin{equation}
    \nonumber
    \mathcal S(u_j \in A) = \sum_{v_i:(u_j,v_i)\in E}\sigma{ji}\cdot e_{ji}\cdot P(v_i|A),
\end{equation}
which can be interpreted as how many suspicious nodes that
user $u_j$ is involved in.
So the user is less suspicious 
if he has a smaller score, with respect to the current 
contrast suspiciousness $\mathcal P$,
where we use $\mathcal P$ to denote a vector of contrast suspiciousness
of all sink nodes.
We build a priority tree to help us efficiently 
find the user with minimum score. 
The priority tree updates
the users' scores and maintains 
the new minimum as the priorities change.
With removing source nodes $A$, the contrast suspiciousness $\mathcal P$ change, 
in which we then update users' scores $\mathcal S$.
The algorithm keeps reducing $A$ until it is empty.
The best $A^*$ maximizing objective $HS$ and $P(v|A^*)$ 
are returned at the end.

\begin{algorithm}[tbhp]
\caption{$GreedyShaving$ Procedure.}
\small
\label{algshaving}
\begin{algorithmic}
    \STATE {\bf Given} bipartite multigraph $\mathcal{G}(U,V,E)$, \\
	\qquad initial source nodes $A_0 \subset U$.
    %\STATE {\bf Output} suspect source nodes $A^*\subset U$, \\
    %	\qquad target sink probabilities $P^*$.
    \STATE Initialize: 
    \bindent 
	\STATE $A=A_0$ \\ 
	\STATE $\mathcal P$= calculate contrast suspiciousness given $A_0$
	    %\mathbf I$, where $\mathbf I$ is an identity vector.
    \eindent
    \STATE $\mathcal S$ = calculate suspiciousness scores of source nodes $A$.
    \STATE $MT$ = build priority tree of $A$ with scores $\mathcal S$.
    \WHILE{$A$ is not empty}
	\STATE $u$ = pop the source node of the minimum score from $MT$.
	\STATE $A = A\setminus{u}$, delete $u$ from $A$.
	%\STATE Delete $u$ from priority tree $MT$.
	\STATE Update $\mathcal P$ with respect to new source nodes $A$.
	%\begin{ALC@g}
	%    \STATE\COMMENT{(takes $O(d_u \cdot (|A|-1))$ time)}
	%\end{ALC@g}
	\STATE Update $MT$ with respect to new  $\mathcal P$.
	%\begin{ALC@g}
	%    \STATE\COMMENT{(takes $O(d_u \cdot (|A|-1)\log|A_0|)$ time)}
	%\end{ALC@g}
	\STATE $HS$ = estimate objective as Equation (\ref{eqobj}).
	%\STATE If $HS > HS^*$, $A^*=A$, $P^*=P$ and $HS^*=HS$  
    \ENDWHILE
    \RETURN $A^*$ that maximizes objective $HS(A^*)$ and $P(v|A^*)$, $v\in V$.
\end{algorithmic}
\end{algorithm}

Since awakening and burst points have been already 
calculated for each sink node as an initial step 
before the $GreedyShaving$ procedure,  
the calculation of the contrast suspicious 
$P(v|A)$ for a sink node $v$ only needs $O(|A|)$ time.
With source node $j$ as the $j$-th one removed from $A_0$ by the
$GreedyShaving$ procedure, $|A_0|=m_0$, and the out degree 
as $d_i$, the complexity is
\begin{eqnarray}
    \label{eqcomplex}
    \sum_{j=2, \cdots, m_0}{O(d_j\cdot(j-1)\cdot \log{m_0})}= O(m_0|E_0|\log{m_0})
    %&=& \log{m_0} \sum_{j=2,\cdots, m_0}{O(d_i\cdot j)} \\
    %&=& O(m_0|E|\log{m_0})
\end{eqnarray}
where $E_0$ is the set of edges connected to source nodes $A_0$.

%\begin{algorithm}
%\caption{$GreedyInflating$ Procedure.}
%%\small
%\label{alginflate}
%\begin{algorithmic}
%    \STATE {\bf Given} bipartite multigraph $\mathcal{G}(U,V,E)$,\\
%	\qquad a queue of source nodes $\tilde U$ with specific order
%    %\STATE {\bf Output} suspect source nodes $A^*\subset U$, and
%    %	   target sink probabilities $P^*$.
%    \STATE Initialize $A=\emptyset$.
%    \WHILE{$\tilde U$ is not empty}
%	\STATE $u$ = pop the top source of $\tilde U$.
%	\STATE $A=\{A, u\}$, add $u$ to set $A$ 
%	\STATE Update sink probabilities $P$ with respect to source nodes $A$
%	\begin{ALC@g}
%	    \STATE\COMMENT{comments: $O(d_u\cdot|A|)$}
%	\end{ALC@g}
%	\STATE $HS$ = estimate objective as Equation (\ref{eqobj})
%    \ENDWHILE
%    \RETURN $A^*$ and $P^*$ that maximize objective $HS$.
%\end{algorithmic}
%\end{algorithm}
%
%With the size of initial queue $\tilde U$ as $m_0$, the complexity 
%is $O(m_0|E_0|)$. %$\sum_{i=1,\cdots,m_0}{O(d_i\cdot i)} = O(m_0|E_0|)$. 
%Without ambiguity, $E_0$ is used as the edges covered by source nodes $\tilde U$.
%%to sink nodes $V$.

With the $GreedyShaving$ procedure,
our scalable algorithm can be designed so as 
to generate candidate suspicious source node sets.
%especially without knowing their suspiciousness.
In our implementation, we use singular
vector decomposition (SVD) for our algorithm.
Each top singular vector 
gives a spectral view of high connectivity
communities. However, those singular vectors
are not associated with suspiciousness scores.
Thus combined with the top singular vectors, 
our fast greedy algorithm is given in Alg.~\ref{algfastgreedy}.

\begin{theorem}[Algorithm complexity]
    In the graph $\mathcal{G}(U,V,E)$,
    given $|V| = O(|U|)$ and $|E|=O(|U|^{\epsilon_0})$,
    the complexity of $FastGreedy$ algorithm is
    subquadratic, i.e., $o(|U|^2)$ in little-o notation,
    if the size of truncated user set $|\tilde U^{(k)}|$ $\leq$
    $|U|^{\sfrac{1}{\epsilon}}$, where $\epsilon>\max\{1.5, \frac{2}{3-\epsilon_0}\}$.
\end{theorem}
\begin{proof}
    The $FastGreedy$ algorithm executes $GreedyShaving$ in a 
    constant iterations. 
    $A_0$ is assigned to $\tilde U^{(k)}$ in every $GreedyShaving$ procedure.
    Then $m_0$ $=$ $|A_0|$ $=$ $|\tilde U^{(k)}|$.
    In the adjacency matrix $\mathbb M$ of the graph, we consider the 
    submatrix $\mathbb M_0$ with $A_0$ as rows and $V$ as columns.
    If the fraudulent dense block is in submatrix $\mathbb{M}_0$,
    then we assume that the block has at most $O(m_0)$ columns. 
    Excluding the dense block, the remaining part of $\mathbb{M}_0$ is assumed
    to have the same density with the whole matrix $\mathbb M$.
    Therefore, the total number of edges in $\mathbb{M}_0$ is 
    \begin{equation}
	\nonumber
%	O(|E_0|) = O(m_0^2 + \frac{m_0\cdot |E|\cdot(|V|-m_0)}{|U|\cdot|V|})
%	 = O( |U|^{\sfrac{2}{\epsilon}} + |U|^{\sfrac{1}{\epsilon}-1} |E| -
%	 |U|^{\sfrac{2}{\epsilon}-2}|E| )
	O(|E_0|) = O(m_0^2 + \frac{m_0\cdot |E|}{|U|})
	 = O(|U|^{\sfrac{2}{\epsilon}} + |U|^{\sfrac{1}{\epsilon}-1} |E|)
    \end{equation}
    Then based on equation (\ref{eqcomplex}), the algorithm complexity is 
	%$O(m_0|E_0|\log m_0)$ $=$ $O((|U|^{\sfrac{3}{\epsilon}} + 
	%|U|^{\sfrac{2}{\epsilon}-1+\epsilon_0})\log|U|)$
    \begin{equation}
	\nonumber
	 O(m_0|E_0|\log m_0)   =   O((|U|^{\sfrac{3}{\epsilon}} + 
	|U|^{\sfrac{2}{\epsilon}-1+\epsilon_0})\log|U|) 
    \end{equation}

    Therefore, if $\epsilon>\max\{1.5, \frac{2}{3-\epsilon_0}\}$, then
    the complexity is subquadratic $o(|U|^2)$. 
\end{proof}
In real life graph, $\epsilon_0\leq 1.6$, so if $\epsilon>1.5$ the 
complexity of $FastGreedy$ algorithm is subquadratic.
Therefore, without loss of performance and efficiency,
we can limit $|\tilde U^{(k)}|\leq |U|^{\sfrac{1}{1.6}}$ for truncating 
an ordered $U$ in the $FastGreedy$ algorithm for a large dataset.

In $FastGreedy$ algorithm for HS-$\alpha$, SVD on 
adjacency matrix $\mathbb M$
is used to generate initial blocks
for the $GreedyShaving$ procedure. Although we can still
use SVD on $\mathbb M$ for HS with holistic attributes,
yet considering attributes of timestamps and rating scores
may bring more benefits.  
Observing that not every combination of \# of stars, timestamps and product ids
has a value in a multi-way tenor representation, 
we can only choose every existing
triplets ($object$, $timestamp$, $\#stars$) as
one column, and $user$ as rows, to form a new matrix.
The above transformation is called the $matricization$ of
a tensor, which outputs a new matrix.
With proper time bins, e.g., one hour or day,
and re-clustering of $\#stars$, the flattening matrix
becomes more dense and contains more attribute information.
Thus we use such a flattening matrix with each column 
weighted by the sudden-drop suspiciousness for our $FastGreedy$ 
algorithm.

\begin{algorithm}[htbp]
\caption{$FastGreedy$ Algorithm for Fraud detection.}
\small
\label{algfastgreedy}
\begin{algorithmic}
    \STATE {\bf Given} bipartite multigraph $\mathcal{G}(U,V,E)$.
    %\STATE {\bf Output} suspect source nodes $A^*\subset U$, and
    %	   the probabilities of being target sinks, $P^*$.
    \STATE $\mathbb L$ = get first several left singular vectors
    %\STATE $\Sigma$ = get first several singular values
    \FORALL{$L^{(k)} \in \mathbb{L}$} 
	\STATE Rank source nodes $U$ decreasingly on $L^{(k)}$ 
	\STATE $\tilde U^{(k)}$ = truncate $u\in U$ when $L^{(k)}_u \leq 
	\frac{1}{\sqrt{|U|}}$
	%\STATE $\tilde U^{(k)}$ = truncate $u\in U$ with $l^{(k)}_u \sqrt{s_k}
	%\geq \theta$.
	%\STATE $GreedyInflating$ on truncated queue $\tilde U^{(k)}$.
	\STATE  $GreedyShaving$ with initial $\tilde U^{(k)}$.
    \ENDFOR
    \RETURN the best $A^*$ with maximized objective $HS(A^*)$, \\
     \qquad and the rank of $v\in V$ by $f_{A^*}(v)\cdot P(v|A^*)$. 
\end{algorithmic}
\end{algorithm}

%% file: 4-exp.tex
\section{Experiments}

\begin{table}[htbp]
  \centering
  \small
  \caption{Data Statistics}
    \begin{tabular}{|l|l|l|l|}
    \hline
    Data Name & {\#nodes} & {\#edges} & {time span} \\
    %\hline
    %synthetic data & 3K$ x $3K & 975K  & -- \\
    %\hline
    %Youtube & 1.13M x 1.13M & 2.99M & -- \\
    %\hline
    %Sina Weibo & 2.75M x 8.08M & 50.1M & Nov 2013 - Dec 2013 \\
    \hline
	BeerAdvocate~\cite{mcauley2013amateurs} & 26.5K x 50.8K & 1.07M & Jan 08 - Nov 11 \\
    \hline
	Yelp  & 686K x 85.3K & 2.68M & Oct 04 - Jul 16 \\
%    \hline
%	Amazon Android App & 1.32M x 61.3K & 2.64M & Jan 2011 - Jul 2011 \\
    \hline
	Amazon Phone \& Acc~\cite{mcauley2013hidden} & 2.26M x 329K & 3.45M & Jan 07 - Jul 14 \\
    \hline
	Amazon Electronics~\cite{mcauley2013hidden} & 4.20M x 476K & 7.82M & Dec 98 - Jul 14 \\
    \hline
	Amazon Grocery~\cite{mcauley2013hidden} & 763K x 165K & 1.29M & Jan 07 - Jul 14 \\
    \hline
	Amazon mix category~\cite{mukherjee2012spotting} & 1.08M x 726K & 2.72M & Jan 04 - Jun 06 \\
    \hline
    \end{tabular}%
  \label{tabdata}%
\end{table}%

% Table generated by Excel2LaTeX from sheet 'sumexp'
\begin{table*}[htbp]
  \centering
  \small
  \begin{threeparttable}
  \caption{Experimental results on real data with injected labels}
    \begin{tabular}{|l|l|c|c|c|c|c|c|c|c|}
    \hline
	\multirow{2}[4]{*}{Data Name} & \multirow{2}[4]{*}{metrics*} & \multicolumn{4}{c|}{source nodes}  & \multicolumn{4}{c|}{sink nodes} \\
	\cline{3-10}      &       & M-Zoom & D-Cube & CrossSpot & HS    & M-Zoom & D-Cube & CrossSpot & HS \\
	\hline
	\multirow{2}[4]{*}{BeerAdvocate} & auc   & 0.7280 & 0.7353 & 0.2259 & \textbf{0.9758} & 0.6221 & 0.6454 & 0.1295 & \textbf{0.9945} \\
	\cline{2-10}      & F$\geq$90\%  & 0.5000 & 0.5000 & --    & \textbf{0.0333} & 0.5000 & 0.5000 & --   & \textbf{0.0333} \\
	\hline
	\multirow{2}[4]{*}{Yelp} & auc   & 0.9019 & 0.9137 & 0.9916 & \textbf{0.9925} & 0.9709 & 0.8863 & 0.0415 & \textbf{0.9950} \\
	\cline{2-10}      & F$\geq$90\%  & 0.2500 & 0.2000 & 0.0200 & \textbf{0.0143} & 0.0250 & 1.0000 & --    & \textbf{0.0100} \\
	\hline
	{Amazon} & auc   & 0.9246 & 0.8042 & 0.0169 & \textbf{0.9691} & 0.9279 & 0.8810 & 0.0515 & \textbf{0.9823} \\
	\cline{2-10} Phone \& Acc & F$\geq$90\%  & 0.1667 & 0.5000 & --  & \textbf{0.0200}$^{\dagger}$ & 0.1429 & 0.1000 & -- & \textbf{0.0200}$^{\dagger}$ \\
	\hline
	{Amazon} & auc   & 0.9141 & 0.9117 & 0.0009 & \textbf{0.9250} & 0.9142 & 0.7868 & 0.0301 & \textbf{0.9385} \\
	\cline{2-10} Electronics & F$\geq$90\%  & 0.2000 & 0.1250 & --    & \textbf{0.1000} & \textbf{0.1000} & 0.5000 & --    & 0.1250 \\
	\hline
	{Amazon} & auc   & 0.8998 & 0.8428 & 0.0058 & \textbf{0.9250} & 0.8756 & 0.8241 & 0.0200 & \textbf{0.9621} \\
	\cline{2-10} Grocery  & F$\geq$90\%  & 0.1667 & 0.5000 & --    & \textbf{0.1000} & 0.1250 & 0.2500 & --    & \textbf{0.1000} \\
	\hline
	{Amazon} & auc   & 0.9001 & 0.8490 & 0.5747 & \textbf{0.9922} & 0.9937 & 0.9346 & 0.0157 & \textbf{0.9950} \\
	\cline{2-10}mix category  & F$\geq$90\%  & 0.2500 & 0.5000 & 0.2000$^{\dagger}$  & \textbf{0.0167} & \textbf{0.0100} & 0.2000 & --    & \textbf{0.0100} \\
    \hline
    \end{tabular}%
    \label{tbexp}%
    \begin{tablenotes}
    \footnotesize
    \item * we use the two metrics: the area under the curve (abbrev as low-case ``auc'') 
	of the accuracy curve as drawn in Fig.~\ref{sfigallcmp}, and the lowest
	$detection$ density that the method can detect in high accuracy($\geq 90\%$).
	\item $^\dagger$ one of the above fraudulent density was not detected in high accuracy.
    \end{tablenotes}
  \end{threeparttable}
\end{table*}%

In the experiments, we only consider the significant
multiple bursts for fluctuated time series of sink nodes.
We keep those awakening-burst point pairs with the 
altitude difference $\Delta c$ at least 50\% of 
the largest altitude difference in the time series.
Table~\ref{tabdata} gives the statistics of our 
six datasets which are publicly available for 
academic research~\footnote{Yelp dataset is from
\url{https://www.yelp.com/dataset_challenge} }.
%including BeerAdvocate data~\cite{mcauley2013amateurs}, 
%Yelp data~\footnote{Yelp dataset is from \url{https://www.yelp.com/dataset_challenge} },
%Amazon review data in categories~\cite{mcauley2013hidden},
%and Amazon reviews with mixed categories~\cite{mukherjee2012spotting}.
Our extensive experiments showed that
the performance was insensitive to scaling base $b$, 
and became very stable when larger than 32.
Hence we choose $b=32$ in the following experiments.

%\subsection{Parameter sensitivity}
%As we have theoretically analyzed the scaling base $b$ of
%function $q(x)$, we test the sensitivity of $b$ on two
%data sets with different fraudulent density, choosing between
%$b=$ 2, $e$, 8, 16, 32, 256. The results are shown in
%Fig.~\ref{figsense}. 
%Generally speaking, our performance is not sensitive
%to $b$, as the empirical results shows that around 
%5\% difference between the choices of $b$.
%As $b$ increases, the performance quickly 
%becomes stable after $b=8$, 
%which agrees with the theoretical analysis that
%larger $b$ helps in shaving unsuspicious users. 
%In our experiments, we use $b=32$.
%
%\begin{figure}
%    \centering
%	\includegraphics[height=1.4in]{sensitiveofb.pdf}
%	\caption{HS quickly becomes stable.}
%	\label{figsense}
%\end{figure}

\subsection{Evaluation on different injection density}
%: one is the synthetic data
%generate by our code, shown in Fig~\ref{sfighyper};
In the experiments, we mimic the fraudsters'
behaviors and randomly choose 200 objects 
that has no more than 100 indegree 
as the fraudsters' customers,
since less popular objects are more likely to 
buy fake ratings.
On the other hand, the fraudulent accounts 
can come from the hijacked user accounts.
Thus we can uniformly sample out a number of users 
from the whole user set as fraudsters.
To test on different fraudulent density,
the number of sampled fraudsters ranges from 200 
to 20000. 
Those fraudsters as a whole randomly 
rate each of the 200 products for 200 times, 
and also create some biased camouflage on other products.
As a results, the fraudulent density ranges from 1.0 to 0.01
for testing.
The rating time was generated for each fraudulent edge:
first randomly choosing a start time between
the earliest and the latest creation time of the existing edges; 
and then plus a randomly and biased time interval 
from the intervals of exiting creation time, 
to mimic the surge of fraudsters' attacks.
Besides, a high rating score, e.g. 4 or 4.5,
is randomly chosen for each fraudulent edge\footnote{
the injection code is also open-sourced for reproducibility}.

%data:
%\begin{itemize}
% \item Synthetic data with hyperbolic block
% \item Microblog from weibo.com.
% \item youtube social network from SNAP.  In the Youtube social network, users form friendship each other and users can create groups which other users can join. 
% \item BeerAdvocate rating and time from SNAP
% \item Amozon categorial review data from UCSD
%    Apps for Android,
%    Cell Phones and Accessories,
%    Electronics,
%    Grocery and Gourmet Food :edge 1287466, users: 
% \item Amozon review data from UIC
%\end{itemize}

Fig.~\ref{sfigallcmp} shows the 
results of HoloScope HS on the BeerAdvocate data.
When the fraudulent density decreases from the right to
the left along the horizontal axis, HS can 
keep as high F-measure on accuracy as more than 80\%
before reaching 0.025 in density, better by far
than the competitors. 
Since HS returns suspiciousness scores for sink nodes, 
we measure their accuracy using AUC (the area under the ROC curve),
where ROC stands for receiver operating characteristic.
For the detection of suspicious objects,
HS achieved more than 0.95 in AUC for all the testing 
injection density, with a majority of tests reaching to 1.0.

In order to give a comparison on all six data sets
with different injection density, we propose to use 
the two metrics: a low-case ``auc'' and the lowest
$detection$ density, 
described in the notes of Table~\ref{tbexp}.
The table reports the fraud detection 
results of our HoloScope (HS) and competitors 
on the six datasets. 
%We use the AUC of the accuracy curve as one of the metrics.
%The accuracy is measured with F-measure in most cases,
%while we use the AUC of the ROC curve 
%for the results of suspicious objects by the HS, 
%since the HS outputs the suspiciousness scores for sink nodes. 
Since the accuracy curve stops at 0.01 (the minimum testing density); 
and we add zero accuracy at zero density, 
the ideal value of auc is 0.995.
The auc on source and sink nodes are reported separately.
As the table suggests, our HS achieved the best auc
among the competitors, and even reached the ideal
auc in two cases.
Since the HS outputs the suspiciousness scores for sink nodes, 
we used the area under the upper-case AUC (similar to F-measure) accuracy
curve along all testing density.
Although it is sometimes unfair to compare 
F-measure with AUC, since the smallest value of AUC is around 0.5,
yet the high auc values can indicate the high F-measure values
on our top suspicious list.

Furthermore, we compare the lowest $detection$
density in Table~\ref{tbexp}.
%i.e. the lowest density for which they achieve accuracy of $\geq 90\%$. 
The better a method is, the lower density 
it should be able to detect well.
As we can see, HS has the smallest detection density
in most cases, which can be as small as $\sfrac{200}{14000}$$=$
0.0143 on source nodes, and reached the minimum testing 
density of 0.01 on sink nodes. That means we can detect fraudsters in high
accuracy even if they use 14 thousand accounts to create
200 $\times$ 200 fake edges for 200 objects. 
The fraudulent objects
can also be detected accurately, even if 20 thousand 
fraudsters are hired to create 200 fake edges for 
each object.

\vspace{-0.2cm}
\subsection{Evaluation on Sina Weibo with real labels}
We also did experiments on a large real dataset from Sina Weibo,
which has 2.75 million users, 8.08 million messages,
and 50.1 million edges in Dec 2013. 
The user names and ids, and message ids are from the online system. 
Thus we can check their existence status in the system to evaluate the
experiments. If the messages or the users were deleted 
from the system, we treat them as the basis 
for identifying suspicious users and messages.
Since it is impossible to check all of the users and messages,
we firstly collected a candidate set, which is the union of
the output sets from the HS and the baseline methods.
The real labels are from the candidate set by checking
the status whether they still exists in 
Sina Weibo (checked in Feb. 2017). We used a program on
the API service of Sina Weibo to check the candidate user and message id lists,
finally resulting in 3957 labeled users and 1615 
labeled messages.

The experimental results in Fig.~\ref{sfigwbexp} show
that HS achieved high F-measure on accuracy, which
detected 3781 labeled users higher than M-Zoom's 
1963 labeled users. The F-measure of HS improved
about 30\% and 60\%, compared with M-Zoom and 
D-Cube respectively.
CrossSpot biased to include a large amount of users in
their detection results, which detected more
than 100 labeled users but with extremely low precision, 
i.e., less than 1\%. That is the reason CrossSpot
got the lowest F-measure, which is less than 1.5\%.
For labeled messages, the HS achieved around 0.8
in AUC, while M-Zoom and D-Cube got lower recall,
and CrossSpot still 
suffered very low F-measure with higher recall. 
Therefore, our HoloScope outperformed 
the competitors in real-labeled data as well.
 
    %Sina Weibo & 2.75M x 8.08M & 50.1M & Nov 2013 - Dec 2013 \\
\subsection{Scalability}
To verify the complexity, 
we choose two representative datasets: BeerAdvocate data
which has the highest volume density, and Amazon Electronics
which has the most edges.
We truncated the two datasets according to different
time ranges, i.e., from the past 3 months, 6 months, or several years
to the last day, so that the generated data size increases.
Our algorithm is implemented in Python.
As shown in Fig.~\ref{figeffe}, 
the running time of our algorithm increases almost linearly with the number of the edges.

%We split the dataset BeerAdvocate according to different
%time ranges, from Jan 2012, Jan 2011, $\cdots$, %Jan 2008, Jan 2006, 
%Oct 2004 to Nov 2011. We then ran our algorithm on
%each resulting dataset. 
%Another larger dataset Amazon Electronics was also 
%tested by different time, ranges from Dec 1998, Jan 2003,
%Jan 2005, $\cdots$, Mar and Jun 2014 to Jul 2014.
%Our algorithm is implemented in Python.
%As shown in Fig.~\ref{figeffe}, 
%the running time of our algorithm increases almost linearly with the number of the edges.

%% file: 5-con.tex
\section{Conclusion}
We proposed a fraud detection method, HoloScope,
on a bipartite graph which can have timestamps
and rating scores. HoloScope has the following advantages:
{1)~\textbf{Unification of signals:}} we make holistic
	use of several signals, namely topology, 
	temporal spikes, and rating deviation
	in our suspiciousness framework in a systematic way.
{2)~\textbf{Theoretical analysis of fraudsters' obstruction:}}
	we showed that if the fraudsters use less than an upper bound
	of time to rate an object, they will 
	cause a suspicious drop or burst. In other words, our
	HoloScope can obstruct fraudsters and increases their time cost.
{3)~\textbf{Effectiveness:}} we achieved higher accuracy 
	on both semi-real and real datasets than the competitors, 
	achieving good accuracy even when the fraudulent density is low.
{4)~\textbf{Scalability:}} while HoloScope needs to dynamically
	update the suspiciousness of objects, the algorithm is
	sub-quadratic in the number of nodes, under reasonable assumptions.